\pdfoutput=1
\documentclass[reqno]{elsarticle}
\usepackage[a4paper, margin=1in]{geometry}

\usepackage{amsmath}
\usepackage{amssymb,amsthm,enumerate}
 \usepackage{color}
 \usepackage[margin=0.5cm]{caption}
\usepackage{hyperref}

\makeatletter
\newtheorem*{rep@theorem}{\rep@title}
\newcommand{\newreptheorem}[2]{%
\newenvironment{rep#1}[1]{%
 \def\rep@title{#2 \ref{##1}}%
 \begin{rep@theorem}}%
 {\end{rep@theorem}}}
\makeatother

\newtheorem{theorem}{Theorem}
\newtheorem{claim}{Claim}
\newreptheorem{claim}{Claim}
\newtheorem{lemma}{Lemma}
\newreptheorem{lemma}{Lemma}
\newtheorem{definition}{Definition}

\newtheorem{corollary}{Corollary}
\newtheorem{remark}{Remark}
\newtheorem{fact}{Fact}
\newreptheorem{fact}{Fact}

\newenvironment{alphaenumerate}{
  \begin{enumerate}[(a)]
}{%
\end{enumerate}
 }

\newenvironment{romanenumerate}{
  \begin{enumerate}[(i)]
}{%
\end{enumerate}
 }

\newenvironment{claimproof}[1][\proofname]{
  \begin{proof}
}{%
\end{proof}
 }
\DeclareMathOperator*{\argmax}{arg\,max}
\newcommand*\diff{\mathop{}\!\mathrm{d}}

\usepackage{mathtools,bbm}

\DeclareMathOperator{\Exp}{E}
\DeclarePairedDelimiter\floor{\lfloor}{\rfloor}
\DeclarePairedDelimiter\ceil{\lceil}{\rceil}
\DeclarePairedDelimiter\round{\lfloor}{\rceil}
\newcommand{\Normal}{\mathcal{N}}
\newcommand{\W}{\mathcal{W}}
\newcommand{\dw}{\mathrm{ans}_w}
\newcommand{\poly}{\mathrm{poly}}
\newcommand{\polylog}{\mathrm{polylog}}
\newcommand{\R}{\mathbb{R}}
\renewcommand{\th}{^\text{th}}

\usepackage[dvipsnames]{xcolor}

\begin{document}

\title{The power of adaptivity in source identification\\ with time queries on the path}

\author[1]{Victor Lecomte}
\ead{vlecomte@stanford.edu}
\author[2]{Gergely Ódor}
\ead{gergely.odor@epfl.ch}
\author[2]{Patrick Thiran} 
\ead{patrick.thiran@epfl.ch}

\address[1]{Stanford University, USA}
\address[2]{EPFL, Switzerland}

\begin{abstract}
We study the problem of identifying the source of a stochastic diffusion process spreading on a graph based on the arrival times of the diffusion at a few queried nodes. In a graph $G=(V,E)$, an unknown source node $v^* \in V$ is drawn uniformly at random, and unknown edge weights $w(e)$ for $e\in E$, representing the propagation delays along the edges, are drawn independently from a Gaussian distribution of mean $1$ and variance $\sigma^2$. An algorithm then attempts to identify $v^*$ by querying nodes $q \in V$ and being told the length of the shortest path between $q$ and $v^*$ in graph $G$ weighted by $w$. We consider two settings: \emph{non-adaptive}, in which all query nodes must be decided in advance, and \emph{adaptive}, in which each query can depend on the results of the previous ones. Both settings are motivated by an application of the problem to epidemic processes (where the source is called patient zero), which we discuss in detail.

We characterize the query complexity when $G$ is an $n$-node path. In the non-adaptive setting, $\Theta(n\sigma^2)$ queries are needed for $\sigma^2 \leq 1$, and $\Theta(n)$ for $\sigma^2 \geq 1$. In the adaptive setting, somewhat surprisingly, only $\Theta(\log\log_{1/\sigma}n)$ are needed when $\sigma^2 \leq 1/2$, and $\Theta(\log \log n)+O_\sigma(1)$ when $\sigma^2 \geq 1/2$. This is the first mathematical study of source identification with time queries in a non-deterministic diffusion process.
\end{abstract}


\maketitle
%mandatory: add short abstract of the document

%\keywords{graph algorithms, source location, noisy information, lower bounds}

% BEGIN \input{intro.tex}
\section{Introduction}
\label{sec:intro}
When a diffusion process spreads in a network, identifying its \textit{source}, i.e., the first node $v^*$ that started the diffusion, is a difficult and intriguing task. Depending on the application, the diffusion process can model a variety of real-world phenomena, including a worm in a computer network \cite{xie2005worm}, a false-rumor in a social network \cite{shah2011rumors}, or an epidemic process \cite{PintoTV12}. In epidemics, the identification of the source (also called patient zero) can be useful while planning our response as a society, since any information on the disease is crucial in uncertain times \cite{ingraham2021omicron} (e.g. source identification can aid contact tracing efforts \cite{carinci2020covid,russo2020tracing}, or it can give information on how dangerous the outbreak is in the case of a new mutation \cite{kandeelomicron,kupferschmidt2021did}). While often useful, we note that the identification of the source may be undesired in certain cases due to privacy concerns \cite{fanti2017deanonymization,prasad2020role}; we refer to \cite{fanti2016metadata,fanti2015spy} for theoretical work addressing this issue. Since the majority of the recent work in source identification is focused on disease spreading, we adopt the language of epidemics in the introduction. 

If we could observe the entire process of the epidemic propagation and know the precise infection times, identifying its source would be easy. Unfortunately, due to the costs of information collection and the overhead constraints, the data available for source identification is usually very sparse. There are two popular frameworks for source identification that mathematically formalize the data-sparsity constraint: in the setting with binary queries (also called snapshot-based setting), proposed by \cite{shah2011rumors}, every node reveals whether they are infected or not at some time $t$, whereas in the setting with time queries (also called sensor-based setting), proposed by \cite{PintoTV12}, a small subset of nodes, which we call hereafter \emph{queries} or \emph{query nodes}, reveal their infection time after the epidemic has spread to the entire network. The two frameworks are quite different, and in this paper we consider only the formulation with time queries, which is driven by the following three research questions of increasing complexity (as identified by \cite{zejnilovic2017sequential}):

\begin{romanenumerate}
\item given the answers to a fixed set of queries, how can we estimate the source?
\item given a the maximum number of queries that we can ask, which queries should we choose so that we can solve the estimation problem as accurately as possible?
\item if we want to correctly identify the source, what is the minimum number of queries that we need to ask? 
\end{romanenumerate}
We note that in theoretical papers, (ii) and (iii) are difficult to separate, however, in applied papers (ii) is often solved before (iii).

The answers to these three research questions depend on the specific assumptions on the epidemic model. The original paper \cite{PintoTV12} assumed that the epidemic spreads on a fixed (and known) network following the Susceptible-Infected dynamics \cite{newman2018networks} (also known as first passage percolation \cite{auffinger201550}) with a known edge-delay distribution to model the randomness in the spread of the epidemic. Additionally, it is assumed that the epidemic already has infected everyone in the network, and therefore every query node can reveal their precise infection time. Finally, in \cite{PintoTV12} it is assumed that query nodes also reveal the neighbor from whom they received the infection. This last assumption relies on information that is difficult to obtain in the context of epidemics, and most follow-up works have dropped it. The resulting modified version of \cite{PintoTV12}, which we call S1, is the most popular model in source identification with time queries, and has been the subject of a long list of papers, which address problem (i) \cite{hu2018localization,li2019locating,paluch2018fast,paluch2020locating,shen2016locating,tang2018estimating,xu2019identifying,zhu2016locating} and problem (ii) (see \cite{paluch2020optimizing}, and the references therein) by algorithms of heuristic nature. Our goal in this paper is to rigorously address problem (iii), albeit on a simpler network model, the path network. The epidemic model in our paper, which we call S2, is exactly the same as S1 with one additional assumption: we assume that the infection time of the source is known. We have several reasons to focus on the S2 model instead of the S1, including that S2 is easier to define and it is theoretically more appealing, as pointed out by several papers in the field \cite{zejnilovic2013network,OdorT19}, and that there is little difference between the number of queries required in the two models \cite{spinelli2018many}. We further discuss the differences between the two models and how our results can be extended to S1 in \ref{appendix:SD1}.

One of the main criticisms of source identification algorithms is that the number of queries required to find the source is large. Although this has not been shown theoretically before our paper, it is widely accepted that source identification is possible only if a constant fraction of the population are queried, which makes the developed algorithms unfit for real-world scenarios. To remedy the situation, a recent research direction suggests to give up the exact identification of the source and to replace it by the computation of confidence sets around it, which can be done with fewer queries (see \cite{bubeck2017finding, khim2016confidence} for the binary query and \cite{dawkins2021diffusion} for the time query settings). However, if our goal is to find the source exactly without querying a prohibitively large fraction of the population, the underlying model needs to be changed. A promising approach is to allow the queries to be selected adaptively to previous answers \cite{zejnilovic2015sequential,zejnilovic2017sequential}, which we call \textit{the adaptive setting}. Adaptive strategies have been studied by Spinelli, Celis and Thiran \cite{spinelli2017general,spinelli2017back} by simulations, and they show a large reduction in the number of required queries in real networks. It is important to quantify the magnitude of the reduction, because it is safe to assume that adaptive queries incur more expensive operational costs than the non-adaptive queries, and it is possible that placing $\Theta(\sqrt{N})$ queries adaptively still remains infeasible in practice. To be self-contained, we include simulations in Figure~\ref{fig:sim_fig}(a), which suggest that in the adaptive setting, the number of required queries grows slowly as a function of the network size, especially on geometric networks. But whether the growth is logarithmic or even lower is difficult to estimate from such plots. In this paper, we show that on the path network, we only need $\Theta(\log\log(N))$ queries, which is practically constant in most applications.

We are aware of only one other theoretical work that addresses the role of adaptivity in source identification \cite{OdorT19}, however, they only consider the case when the propagation delays are deterministic. In this case, if the first infection time is also known (model S2), problem (i) is trivial, problem (ii) is equivalent to finding a \emph{resolving set} in a graph (a set of nodes such that the distance to those nodes is enough to uniquely determine the identity of an unknown node) \cite{zejnilovic2013network}, and problem (iii) is equivalent to the \emph{metric dimension} problem \cite{slater1975leaves}. If the time when the infection starts is unknown (model S1), the corresponding combinatorial notion is the \emph{double metric dimension} \cite{chen2014approximability}. In the past few years, there has been a line of work on the metric and the double metric dimensions in the context of source identification of both of simulation-based \cite{spinelli2017general,spinelli2017effect} and of rigorous nature \cite{spinelli2018many}. In fact, the adaptive setting in the context of source identification was also first proposed with the deterministic propagation delay assumption \cite{zejnilovic2015sequential}. The adaptive (sequential) version of the metric dimension also exists in the combinatorial literature \cite{seager2013sequential}, and the result of Odor and Thiran \cite{OdorT19} says that in Erd\H{o}s-R\'enyi graphs, the difference between the adaptive and non-adaptive settings is only a constant factor, which suggests that adaptivity plays little role in this setting. On the other hand, Kim el. al. \cite{kim2015identifying} finds that the sequential metric dimension is $O(\Delta \log(n))$ on graphs with maximum degree $\Delta$ (Theorem 1.2), which is relatively low compared to the metric dimension of most random tree distributions, which tends to be $\Theta(n)$ \cite{mitsche2015limiting,komjathy2021}, suggesting a large role of adaptivity in these settings. Unfortunately, neither Erd\H{o}s-R\'enyi graphs, nor tree graphs are good models of real networks, which motivates the analysis of further network models. A recent work by Lichev, Mitsche and Pra\l{}at \cite{lichev2021localization} showed that the metric dimension of random geometric graphs of $n$ nodes in the unit square with connectivity range $r$ is $\Omega(\mathrm{max}(1/r^2, n^{2/3}r^{4/3}/\log^{1/3}(n)))$ for $1/\sqrt{n} \ll r \le 1/4$ (Theorem 5.2), which together with the upper bound on the sequential metric dimension by Kim el. al. \cite{kim2015identifying}, and with the result that the maximum degree of random geometric graphs is $O(nr^2)$ with high probability \cite{penrose2003random} suggests a large (polynomial in $n$) role of adaptivity. Simpler geometric graph models without randomness, such as the path graph and the grid graph do not exhibit such a big difference between the adaptive and non-adaptive settings (the metric dimension of the path and the grid is 1 and 2, respectively). It is an interesting question to investigate whether adaptivity plays an important role in these simple geometric graph models if randomness is introduced into the diffusion process (as in the S1 or S2 models) and not in the graph model, especially since in most applications of the source identification problem, the diffusion process is in fact stochastic. In Figure~\ref{fig:sim_fig}(b) we show that in the adaptive setting, the number of queries required to find the source depend on the stochasticity of the diffusion, as one would expect; the larger the stochasticity, the more queries are needed. Few works study in more detail the role of stochasticity in source identification, even empirically and in the non-adaptive setting alone (see e.g., \cite{spinelli2017effect}), and we are not aware of any previous work that has determined the exact dependence of the number of required queries on the randomness of the epidemic, neither in the adaptive nor in the non-adaptive setting.

 \begin{figure}
\begin{center}
  \includegraphics[width=0.9\textwidth]{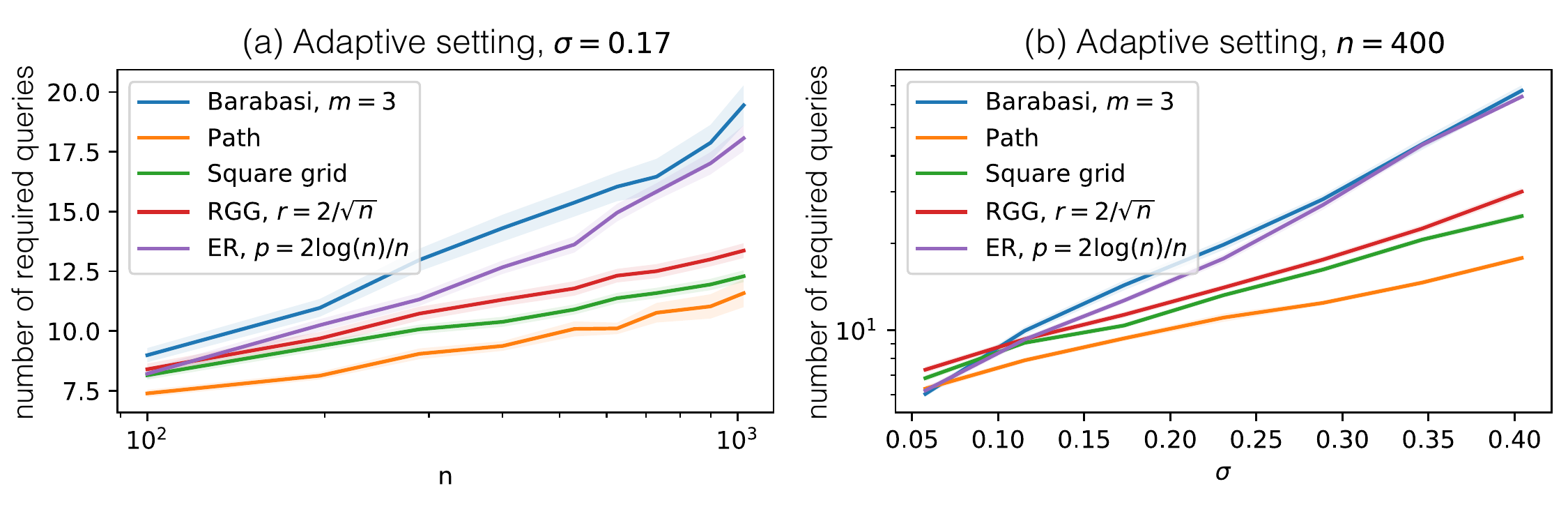}
  \caption{The number of queries required by the adaptive source identification algorithm called Max-Gain \cite{spinelli2017back} in the S1 source identification model as a function of (a) the network size $n$, and (b) the standard deviation of the edge-delays $\sigma$ in the following network models: Barab\'asi-Albert network \cite{barabasi1999emergence} with average degree $2m$, path graph, square grid, random geometric graphs \cite{penrose2003random} in the unit square with connection radius $r$, Erd\H{o}s-R\'enyi graphs \cite{erdHos1959renyi} with connection probability $p$. The network parameters, controlling the number of edges in the random networks, were chosen slightly above the connectivity threshold, so that all networks in the simulation were connected. Each datapoint is an average of 192 simulations and the confidence intervals are computed using the Student's t distribution-test. See \ref{appendix:sim_fig} for more details about how the Max-Gain was run to produce these simulation results.}
  \label{fig:sim_fig}
\end{center}
\end{figure}

In this paper, we compute the query complexity of the stochastic version of the source identification problem in the S2 model, on an $n$-node path, in both the adaptive and the non-adaptive settings. We chose the path network for its simplicity and for the insight that it may offer into the query complexity when the contact graph $G$ has an underlying geometry (see Figure \ref{fig:sim_fig}). The choice of the S2 model is motivated by its close resemblance to the S1 model (see \ref{appendix:SD1}), which is one of the standard models in the literature, but we mention that both the S1 and the S2 are abstract theoretical models, and they may have several currently unidentified applications. The propagation delays, which we call edge \emph{weights}, are chosen to be i.i.d Gaussian variables with unit mean and variance $\sigma^2$, following the model proposed in \cite{PintoTV12}. One should note that the weights can take negative values, especially when $\sigma$ is large, in contradiction with the non-negativity of propagation delays. Letting weights take negative values further accounts for the randomness in the incubation and reporting times. It makes source identification more challenging because of the absence of a deterministic, monotone dependence between the time of infection of a query node and its distance to the source. We also discuss how to extend our results to different propagation delay distributions, including ones that only take positive values, in \ref{sec:other-distributions}.

We find that for a wide range of $\sigma$, there is a drastic decrease in the number of required queries in the adaptive setting compared to the non-adaptive setting. For constant $\sigma$, which might be the most relevant range for practical purposes, the number of required queries is $\Theta(n)$ in the non-adaptive setting and $\Theta(\log\log n)$ in the adaptive setting. For the more precise dependence of our results on the standard deviation~$\sigma$, we refer to Section~\ref{sec:results}.

\subsection{Related work in Information Theory} The role of adaptivity is a central question in several fields in computer science, including property testing \cite{canonne2018adaptivity}, information theory \cite{chiu2016sequential,lalitha2017measurement} and learning theory \cite{settles2009active}. The most well-known example is perhaps binary search on a line, where being adaptive reduces the number of queries from $n$ to $\log_2(n)$. Such a significant decrease in the query complexity of standard binary search is possible because the queries are very constrained; we can only ask whether the target is to the left or to the right of the queried vertex. If instead we are allowed to query any subset for containment of the target without any noise, then there is no difference between the adaptive and the non-adaptive query complexities (this is the well-known BarKochba or 20 Questions Game between two players, where the first player comes up with an item that the other player must identify by asking (in principle up to 20) yes-no questions). Indeed, $\log_2(n)$ questions are necessary because every answer carries only binary information, and the target can be found by $\log_2(n)$ non-adaptive questions by querying each digit of the binary representation of the index of the target vertex. One way to reintroduce a difference between the adaptive and non-adaptive cases in the 20 Questions Game is to corrupt the answers by a query dependent noise, which was proposed initially by R\'enyi \cite{renyi1961problem}, and has been studied by several follow-up works, including \cite{chiu2016sequential,lalitha2017measurement,zhou2021resolution}.

The problem setup of \cite{lalitha2017measurement} has a close resemblance to our setup. In both cases, the search is done on a line, and the answers are corrupted by Gaussian noise, the variance of which depends on how close the query was to finding the target. The notable differences between \cite{lalitha2017measurement} and our setup are that:
\begin{romanenumerate}
\item we have more restrictive queries (one query in our setup is a single node (hence there are $n$ possible queries), whereas one query in the setup of \cite{lalitha2017measurement} is a subset of the nodes (hence there are $2^n$ possible queries))
\item we receive more information (we receive a noisy version of the distance between the queried vertex and the source, whereas in \cite{lalitha2017measurement} they receive a noisy binary answer for belonging to the query set)
\item in our case the noise that corrupts the answers is not independent between queries. 
\end{romanenumerate}

Because of these differences, our proof techniques and our results are also different from \cite{lalitha2017measurement}. The main tool in \cite{lalitha2017measurement} for the adaptive upper bound is the \textit{posterior matching scheme}. Roughly speaking, posterior matching produces queries that split the line into two approximately equal-weight subsets weighed by the posterior. In particular, there is no restriction on the queries produced by the posterior matching scheme, and therefore it is not applicable in our case (see (i) above). We also note, that as opposed to our setup, in \cite{lalitha2017measurement}, the geometry of the search space does not play an important role; any subset of the vertices of the line can be a query and the answers are insensitive to the distances between the queried vertices and the target vertex. For this reason, the usual geometry-insensitive information-theoretic notions (such as the entropy of the posterior) that work well in \cite{lalitha2017measurement}, cannot be used in our setup (see Section \ref{sss:mu} for our notion of ``progress''). In terms of results, for constant $\sigma$, both the non-adaptive and adaptive query complexities are found to be $\Theta(\log(n))$ in \cite{lalitha2017measurement}, which is in sharp contrast with our finding of $\Theta(n)$ and $\Theta(\log\log n)$ in the non-adaptive and adaptive settings, respectively. Finally, we note that the paper \cite{lalitha2017measurement} features results about the expected query complexity of the search algorithms, whereas we give query complexity bounds that hold with any constant failure (or success) probability.

\subsection{Related Work in Theoretical Computer Science} Extensions of binary search to graphs have been proposed on numerous occasions \cite{EmamjomehKS16,FeigeRPU94, KarpK07, OnakP96}. Of these, perhaps \cite{EmamjomehKS16} has the closest connections with source identification with time queries. In this extension, a target vertex at an unknown position in a general graph is to be identified by adaptively querying vertices. A queried vertex can only respond whether it is the target or not, and if not, it indicates the edge on a shortest path between itself and the target. In the noiseless setting, queries always report the correct answer, whereas in the noisy setting, queries report a correct answer independently with probability $1/2 < p < 1$. In a sense, noisy binary search is an adaptive version of the source identification model proposed by \cite{PintoTV12} where we would keep the ``who infected me'' information and drop the time information instead, with the notable difference that in noisy binary search the noise that corrupts the answers is independent between queries. Since the information that a queried node can provide is its distance and/or its direction towards the source, adaptive source identification and noisy binary search on a line can be seen as ``duals'' of each other, in the sense that the former collects a noisy estimate of the distance whereas the latter collects a noisy estimate of the direction to the source. In the latter case, the adaptive query complexity is found to be $\Theta(\log n)$ for constant $p$ in \cite{EmamjomehKS16,lalitha2017measurement}. Comparing this result with our result of $\Theta(\log\log n)$ for the number of required queries in the adaptive case indicates that the distance to the target is far more informative than the direction, at least on the path graph. On different graphs, notably on star graphs, the distance is expected to be less informative than the direction. We limit the study of stochastic source identification in this paper to the path topology because of the complexity of the computations, and we leave the study in other graph topologies for further work.

%\subsection{Related work in learning theory} 
%%Recently, there has been an interest in understanding the tradeoff between privacy and query complexity in binary search.
%
%In the Private Sequential Learning problem, our task is to perform adaptive binary search on a line while an adversary is watching our queries, but not the (in this case binary) responses \cite{tsitsiklis2018private,xu2018query}. We are interested in the minimum number of queries that are necessary, so that at the end of the protocol we know where the target is, but the adversary does not. In a follow-up work the model has been extended to the case when the answers we receive are noisy \cite{xu2019optimal}. In our paper, we do not consider privacy preserving learning, however, we mention it for completeness as the problem setup bears some similarities.
% END \input{intro.tex}
% BEGIN \input{model.tex}
\section{The model}

% Precise definition of the model for general graphs (static and sequential)
A known graph $G = (V,E)$ is fixed in advance. 
First, a node $v^* \in V$ is picked uniformly at random, which is called the \emph{source}. Then for each edge $e\in E$, a weight $w(e)$ is drawn independently from some distribution $\W$.
Both $v^*$ and the weights $w(e)$ are hidden from the identification algorithm.
Once they are drawn, the algorithm will start making queries to identify $v^*$.
To perform a query, the algorithm chooses a query node $q\in V$, and receives an \emph{answer} with value $\dw(v^*,q)$: the shortest distance between $v^*$ and $q$ in graph $G$ with edges weighted by $w$. 

 \begin{figure}
\begin{center}
  \includegraphics[width=\textwidth]{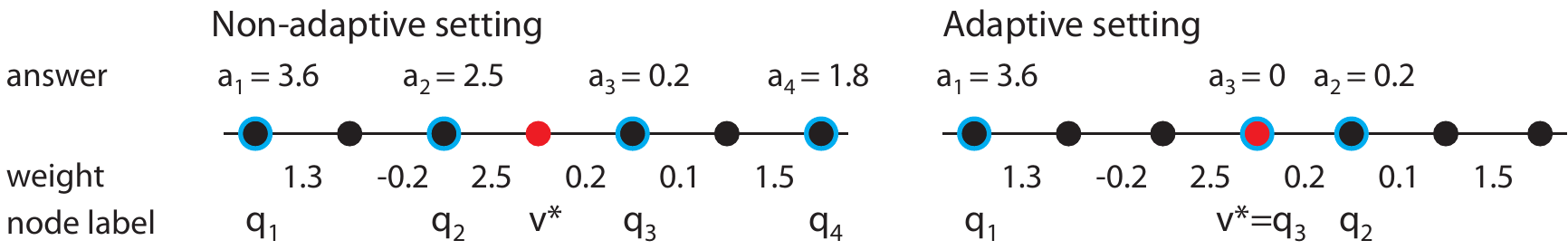}
  \caption{An illustration of the model. Query nodes are marked blue and the source is marked red.}
  %The circles with blue stroke correspond to the sensors and the red circle signifies the location of the source.}
  \label{illustration}
\end{center}
\end{figure}

We distinguish between two settings. In the \emph{non-adaptive} setting, the algorithm has to submit all of its queries in one batch, then receives all answers, and has to make a prediction. In the \emph{adaptive} setting, the algorithm can make queries one by one, and adapt the choice of the next query based on previous answers. In both settings, the weights $w(e)$ are only drawn once, at the very beginning, and will not change between queries. As we will see, the difference between these two settings will have a huge impact on the number of queries that the algorithm needs, because in the adaptive setting the algorithm will be able to quickly zero in on the source $v^*$ and receive progressively more refined information.

In this paper, we treat the case where $G$ is an $n$-node path, with nodes numbered from 1 to $n$ (so $V = \{1,\ldots,n\}$). We will assume $n \geq 3$ for convenience. We will often say that a node $u$ is to the ``left'' (respectively, ``right'') of a node $v$ if $u<v$ (resp., $u>v$). For the weight distribution, we choose $\W=\Normal(1,\sigma^2)$: a normal of mean 1 and variance $\sigma^2 > 0$, where $\sigma$ is a parameter of the model. We choose a normal distribution because (i) by the Central Limit Theorem, the distances between faraway nodes converge to a normal distribution for most edge-delay distributions $\W$, while close-by nodes can be searched via exhaustive search (we discuss how to extend our results to these other distributions in~\ref{sec:other-distributions}), (ii) there are several properties (e.g. additivity, tight concentration) of the normal distribution that simplify our calculations. 

%Formally, we will say that source location can be done in $q(n,\sigma)$ queries if for any fixed probability of failure (resp., success) $\delta$, there is an (resp., there is no) algorithm which can identify $v^*$ with probability $\geq 1-\delta$ while making $O_\delta(q(n,\sigma))$ queries. Our upper bounds will be efficient algorithms, while our lower bounds will be information-theoretic.
% END \input{model.tex}
% BEGIN \input{results.tex}
\section{Results and Discussion}
\label{sec:results}

\subsection{Non-Adaptive Setting}

We present matching upper and lower bounds for the non-adaptive setting. 
\begin{theorem}
\label{thm:static-ub}
For any failure probability $0<\delta<1/2$, there is a deterministic algorithm for non-adaptive source identification on the $n$-node path which asks $\min(O(1+n\sigma^2\log(1/\delta)), n)$ queries and identifies $v^*$ correctly with probability at least $1-\delta$, even if $v^*$ is chosen adversarially instead of drawn uniformly at random.
\end{theorem}

\begin{theorem}
\label{thm:static-lb}
For any success probability $p>1/n$, any (potentially randomized\footnote{Since the distribution for the identity of the source is fixed, rather than adversarial, randomness in the algorithm is not useful (as long as we are not considering running time): the algorithm should simply choose the set of queries that maximizes the probability of finding the source, and output the likeliest source given the answers it receives.}) algorithm for non-adaptive source identification on the $n$-node path
must ask $\Omega(1+\min(p^3n\sigma^2, p^2n)) = \Omega_p(1+\min(n\sigma^2, n))$ queries to identify $v^*$ correctly with probability at least $p$ when $v^*$ is drawn uniformly at random.

%\footnote{Probability $0.99$ was picked to make the proof simpler to understand. It is possible to strengthen Theorem~\ref{thm:static-lb} by relaxing the condition to $\delta > 0$ and changing $\Omega(1+n \cdot \min(\sigma^2, 1))$ to $\Omega_\delta(1+n \cdot \min(\sigma^2, 1))$. In any case, the current formulation is enough to determine the query complexity $q(n,\sigma)$.}%, but the proof becomes significantly longer without being more insightful. }
\end{theorem}

We can interpret the results in the following way. Intuitively, answers are ``accurate'' up to distance roughly $1/\sigma^2$: indeed, for a query node at distance $d$ from the source, the mean of the received answer is $d$ and the variance is $d\sigma^2$, so if $d= \omega(1/\sigma^2)$, the variance becomes $\omega(1)$ and it becomes impossible to deduce the real distance with constant probability. Therefore, instead of thinking of receiving stochastic answers, we can imagine that we receive the exact distance, but only if this distance is $\leq 1/\sigma^2$ (and otherwise they would give no answer at all). That is, we think of the queries as effective within a ``limited range'' $1/\sigma^2$. In that model, it is clear that $\Theta(1+\min(n\sigma^2, n))$ queries are necessary and sufficient, which is exactly what we find in Theorems~\ref{thm:static-ub} and~\ref{thm:static-lb}.

Our proofs also build on the intuition of query nodes with limited range. In the proof of Theorem \ref{thm:static-ub}, we show that if the query nodes are spaced equally (and deterministically), every $v^*$ is in the ``range'' of the closest two query nodes, meaning that once rounded to the closest integer, they give the correct answer with high probability.
This probability is based only on the randomness of the edge weights $w(\cdot)$, and is high no matter where $v^*$ ends up, hence we can identify $v^*$ even without assuming any prior on its distribution. This extra guarantee was not required by the model, but it comes naturally without any additional cost.

In contrast, in Theorem \ref{thm:static-lb} we show that any algorithm that succeeds with constant probability must use $\Omega(1+\min(n\sigma^2, n))$ queries even if the algorithm is allowed to take advantage of the assumption that $v^*$ is uniformly distributed over the nodes $V$. The proof works by showing that if one uses fewer queries, then most of the nodes are so far away from the closest query node that they are indistinguishable from other close-by nodes.

While in this paper we only consider the path graph, these results indicate that limited range query nodes might be a good proxy for non-adaptive source identification in other graphs as well, which has not been thoroughly explored in the source identification literature.

\subsection{Adaptive Setting}

The adaptive setting is more complex and more interesting than the non-adaptive case.
For instance, it is not obvious anymore how the queries should be selected.
We may consider an algorithm that, at each decision, selects a query node based on the posterior probabilities that the source is at some node: we call \emph{posterior} at a node $v$ the probability that the source is $v$, conditioned on the answers made so far. However, those posteriors might be hard to compute, and might not be well-behaved as a function of $v$ (for example, they might not be unimodal).
Fortunately, as long as the variance of the edge-delays is relatively low, the answers that we see are concentrated around their expected value and we can form a fairly good idea about what the posteriors might look like.
This inspires the following algorithm, which (for intuition) can be seen as a procedure that computes at each step the posteriors approximately, and selects the next query node close to the node with the highest posterior (the node that is most likely to be the source).

\begin{theorem}
\label{thm:sequential-ub}
For any failure probability $0<\delta<1/2$, there is a deterministic algorithm for adaptive source identification on the $n$-node path which uses
\[
\begin{cases}
O(\log(1+\log_{1/\sigma} n) + \polylog(1/\delta))\text{ queries if $\sigma^2 \leq 1/2$}\\
O(\log \log n + \sigma^2\cdot \polylog(\sigma, 1/\delta))\text{ queries if $\sigma^2 \geq 1/2$}
\end{cases}
\]
and identifies $v^*$ correctly with probability at least $1-\delta$, even if $v^*$ is chosen adversarially instead of drawn uniformly at random.
\end{theorem}

To show optimality, as we did in the non-adaptive case, we show that no algorithm can succeed without asking a large number of queries, even under the assumption that $v^*$ is uniformly distributed over $V$.

\begin{theorem}
\label{thm:sequential-lb}
For any success probability $p>1/n$ and any $n \geq \Theta_p(\max(\sigma^3, 1))$, any (potentially randomized) algorithm for adaptive source identification on the $n$-node path must use
\[
\begin{cases}
\Omega_{p}(1 + \log (1+\log_{1/\sigma} n))\text{ queries if $\sigma^2 \leq 1/2$}\\
\Omega_{p}(\log \log n)\text{ queries if $\sigma^2 \geq 1/2$}
\end{cases}
\]
to identify $v^*$ correctly with probability at least $p$ when $v^*$ is drawn uniformly at random.
\end{theorem}

\begin{figure}[h]
\begin{center}
  \includegraphics[width=0.8\textwidth]{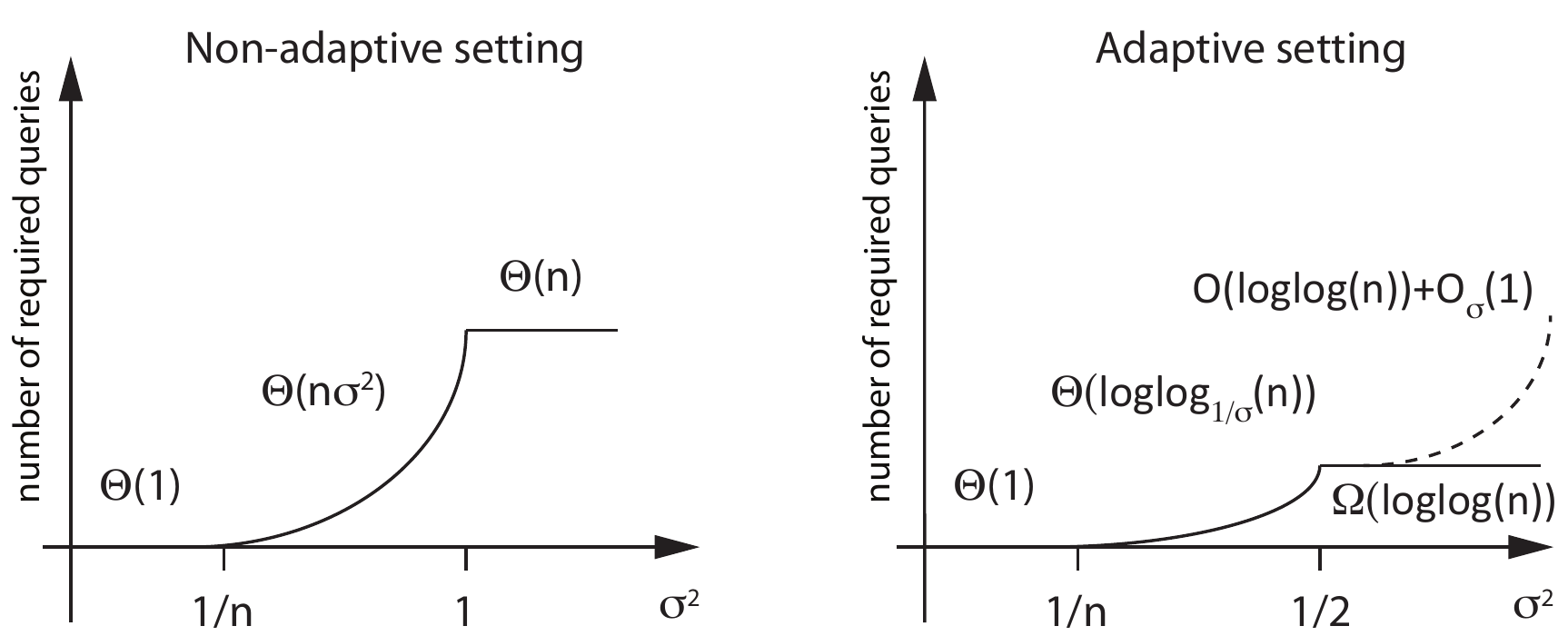}
  \caption{A sketch of a linear-log plot of the optimal number of queries in the non-adaptive and adaptive cases as a function of $\sigma$.}
  \label{plots}
\end{center}
\end{figure}

The proof of Theorem \ref{thm:sequential-lb} is the most challenging proof we present. Since the algorithm is allowed to adapt each query based on the results of the previous ones, we have to somehow quantify the progress it has made towards identifying $v^*$. However, this is made delicate by the fact that the edge weights $w(e)$ are drawn only once at the onset, which means that the algorithm does not only accumulate information about $v^*$, but also about the edge weights $w(e)$. In particular, proof approaches that try to ``fool'' the algorithm by giving it answers from a modified distribution tend to fail because the algorithm can test for consistency across queries. Instead, the proof that we present builds on a detailed understanding what the posteriors can look like in each step.

The upper and lower bounds in Theorems~\ref{thm:sequential-ub} and \ref{thm:sequential-lb} match for $\sigma$ up to $\tilde{\Theta}(\log \log n)$. For $\sigma \gg \log\log(n)$, our upper and lower bounds are separated by a $\sigma^2 \cdot \polylog(\sigma)$ term, which comes from the final steps, when the algorithm has gotten so close to the source that the variance of the edge-delays is about as big as the expected value of the answers, at which point the algorithm simply queries every node.

This $\sigma \gg \log\log(n)$ regime is a difficult regime to analyse, because for larger $\sigma$ we lose the concentration of the answers, and cannot control the shape of the posteriors anymore.
Moreover, we believe that in the high-$\sigma$ regime, any asymptotically tight results for the Gaussian case would not carry over to other edge-delay distributions $\W$, because the $O_{\sigma,\delta}(1)$ term becomes sensitive to the specific $\W$ we pick. As an example, consider $\W$ to be a Gaussian distribution with a very large $\sigma$ (say, larger than $n$), truncated at 0 (to prevent negative edge weights $w$). With this $\W$, we can always figure out which direction the source is from node $v$ by simply querying a neighbor of $v$, and checking which answer is larger. Hence we can find the source with binary search in $\log_2(n)$ rounds, however, if $\W$ is a non-truncated Gaussian random variable and $\sigma$ is large enough, we clearly have to query every node to find the source.

Finally, as an additional strength of our results, we note that the lower bounds (Theorems~\ref{thm:static-lb} and~\ref{thm:sequential-lb}) continue to apply even if the algorithm were to also receive direction information (i.e. whether the source is on the left or on the right of the query node), whereas the upper bounds (Theorems~\ref{thm:static-ub} and~\ref{thm:sequential-ub}) work well even without using this information.% END \input{results.tex}
% BEGIN \input{preliminaries.tex}
\section{Preliminaries}

We denote a normal distribution with mean $\mu$ and variance $\sigma^2$ as $\Normal(\mu, \sigma^2)$. We occasionally call variables distributed according to this a normal distribution ``Gaussians''. We will often use the following basic facts about the normal distribution.

\begin{fact}
\label{fact:sum-normal}
If $X \sim \Normal(\mu_1, \sigma_1^2), Y\sim \Normal(\mu_2, \sigma_2^2)$ are independent Gaussians, then $X+Y \sim \Normal(\mu_1+\mu_2, \sigma_1^2+\sigma_2^2)$.
\end{fact}

\begin{fact}
\label{fact:normal}
If $X \sim \Normal(\mu, \sigma^2)$, then $\Pr[X \notin \mu \pm a] \leq e^{-\frac{a^2}{2\sigma^2}}$.
\end{fact}

% BEGIN \input{proofs/normal.tex}
%\subsection{Proof of Fact~\ref{fact:normal}}
%\label{proof:fact:normal}

%\begin{repfact}{fact:normal}
%If $X \sim \Normal(\mu, \sigma^2)$, then $\Pr[X \notin \mu \pm a] \leq %e^{-\frac{a^2}{2\sigma^2}}$.
%\end{repfact}

\begin{proof}%[Proof of Fact~\ref{fact:normal}]
\newcommand{\erf}{\mathrm{erf}}
First, using the probability density function of the normal distribution and the change of variables $z=\frac{x-\mu}{\sigma}$ we have 
\[
\Pr[X \notin \mu \pm a] = 2 \times \int_{\mu+a}^\infty \frac{1}{\sigma\sqrt{2\pi}} e^{-\frac{1}{2}\left(\frac{x-\mu}{\sigma}\right)^2}dx = \sqrt{\frac{2}{\pi}} \int_{a/\sigma}^\infty e^{-z^2/2} dz.
\]
It only remains is to prove that for all $b \geq 0$, $\sqrt{\frac{2}{\pi}}\int_b^\infty e^{-z^2/2}dz \leq e^{-b^2/2}$.

We separate into two cases.
\begin{itemize}
\item If $b \geq 1$, then we have
%\begin{align*}
%1-\erf(x)
%&= \frac{2}{\sqrt{\pi}}\int_x^\infty e^{-t^2}dt\\
%&= \frac{2}{\sqrt{\pi}}\int_x^\infty e^{-t^2}dt\\
%&\leq \frac{2}{2x\sqrt{\pi}}\int_x^\infty 2t\,e^{-t^2}dt\\
%&= \frac{e^{-x^2}}{x\sqrt{\pi}} \leq \frac{e^{-x^2}}{\sqrt{\pi/2}} \leq e^{-x^2}.
%\end{align*}
\begin{align*}
    \sqrt{\frac{2}{\pi}}\int_b^\infty e^{-z^2/2}dz
    &\leq \sqrt{\frac{2}{\pi}}\int_b^\infty \frac{z}{b} e^{-z^2/2}dz\tag{$z\geq b$}\\
    &= \sqrt{\frac{2}{\pi}}\frac{1}{b}e^{-b^2/2}\tag{$\frac{d}{dz}e^{-z^2/2} = -ze^{z^2/2}$}\\
    &\leq e^{-b^2/2}.\tag{$b \geq 1 \geq \sqrt{\frac{2}{\pi}}$}
\end{align*}
\item One can easily check that on interval $[0,1]$, $\sqrt{\frac{2}{\pi}}\int_b^\infty e^{-z^2/2}dz$ is convex while $\sqrt{\frac{2}{\pi}}\int_b^\infty e^{-z^2/2}dz$ is concave (by computing their second derivatives), and in addition we have
%\[\erf(0) = 1-e^{-0^2}\qquad\erf\left(\frac{1}{\sqrt{2}}\right) >
%68\% > 39\% >
%1-e^{-\left(\frac{1}{\sqrt{2}}\right)^2}.\]
\[
\sqrt{\frac{2}{\pi}}\int_0^\infty e^{-z^2/2}dz = 1 = e^{-0^2/2}
\quad\text{and}\quad
\sqrt{\frac{2}{\pi}}\int_1^\infty e^{-z^2/2}dz < 0.318 < 0.606 < e^{-1^2/2}.
\]
Therefore, $\sqrt{\frac{2}{\pi}}\int_b^\infty e^{-z^2/2}dz \leq e^{-b^2/2}$ on the whole interval.\qedhere
\end{itemize}
\end{proof}
% END \input{proofs/normal.tex}
% END \input{preliminaries.tex}
% BEGIN \input{static.tex}
\section{Proofs for the Non-Adaptive Setting}
%In this section, we prove Theorems~\ref{thm:static-ub} and~\ref{thm:static-lb}. This shows that the query complexity in the static setting is $\Theta(1+n\cdot\min(\sigma^2,1))$. That is, $\Theta(1)$ sensors are needed when the variance is very low ($\sigma^2 \leq 1/n$), $\Theta(n\sigma^2)$ sensors are needed when $1/n \leq \sigma^2 \leq 1$, and $\Theta(n)$ sensors are needed when the variance is high ($\sigma^2 \geq 1$).

% BEGIN \input{static-ub.tex}
\subsection{Upper Bound}
\newcommand{\sSmall}{{q_\text{smallest}}}
\newcommand{\sLeft}{{q_\text{left}}}
\newcommand{\sRight}{{q_\text{right}}}
\newcommand{\oSmall}{{a_\text{smallest}}}
\newcommand{\oLeft}{{a_\text{left}}}
\newcommand{\oRight}{{a_\text{right}}}

\begin{proof}[Proof of Theorem~\ref{thm:static-ub}]
First of all, observe that one can always find the source with probability 1 if one is willing to use $n$ queries: just query each node. The query at node $v^*$ will produce answer 0, while the other queries will almost surely produce nonzero answers.

Now it suffices to show that the source can be identified with $O(1+n\sigma^2\log(1/\delta))$ queries. The strategy is natural: query $\sim n/d$ nodes along the path at fixed intervals of some length $d$, where $d$ is small enough to ensure that that the query nodes nearest to the source $v^*$ return an answer that is very close to the real distance (and in particular, that will be exactly equal to it once rounded to the nearest integer).

What makes things a bit more complex is that:
\begin{alphaenumerate}
\item even if all the weights are positive, it may not be easy to determine between which two query nodes $v^*$ is identified;
\item since the weights may be negative, it is possible that a query node $q_1$ gives a smaller answer than a query node $q_2$ even though $q_2$ lies between $q_1$ and $v^*$ (in particular, the answers do not necessarily form a unimodal sequence when read from left to right).
\end{alphaenumerate}

Concretely, the algorithm will query nodes $1,d+1,2d+1,\ldots,1+\floor{\frac{n-1}{d}}d$. It will then find the query node with the smallest answer, which we call $\sSmall$, and the next query node to its left $\sLeft \coloneqq \sSmall-d$ (let's assume for now that $\sSmall \neq 1$, so that $\sLeft$ exists).
Let $\oSmall \coloneqq \dw(v^*,\sSmall)$ and $\oLeft \coloneqq \dw(v^*,\sLeft)$ be the corresponding answers.
Then the algorithm just assumes that both of them are correct (equal to the real distance) once rounded to the nearest integer (that is, $\round{\oSmall} = |v^*-\sSmall|$ and $\round{\oLeft} = |v^*-\sLeft|$), and computes $v^*$ as
\[
\begin{cases}
\sSmall+\round{\oSmall}\text{ if $\round{\oLeft} \geq d$}\\
\sSmall-\round{\oSmall}\text{ otherwise.\footnotemark}
\end{cases}
\]
\footnotetext{If $\sLeft$ does not exist, which happens only when $\sSmall=1$, then the algorithm
%can simulate placing a sensor at ``node $1-d$'' by generating $d$ new weights from $\Normal(1,\sigma^2)$ and adding this to the observation at node $1$.
can simply compute $v^*$ as $1+\round{\oSmall}$, again assuming that $\oSmall$ is correct once rounded to the nearest integer.}

For this strategy to work, it is enough if the following statements hold simultaneously:
\begin{alphaenumerate}
\item among the query nodes located at or to the left of $v^*$, the closest one is the one with the smallest answer;
\item among the query nodes located at or to the right of $v^*$, the closest one is the one with the smallest answer;
\item the two closest query nodes to $v^*$ on its left side and the closest query node on its right side all give a correct answer once rounded to the nearest integer.
\end{alphaenumerate}
Indeed, if this is true, then $\sSmall$ will be the closest query node to $v^*$ on either its left or right side, and thus both $\sSmall$ and $\sLeft$ will be among the three query nodes that are guaranteed by point (c) to give the correct result once rounded.

The following claim, which is purely technical and easily obtained from concentration bounds, is proved in~\ref{proof:claim:end-of-static-ub}.
\begin{claim}
\label{claim:end-of-static-ub}
For some $d = \Omega\left(\frac{1}{\sigma^2\log(1/\delta)}\right)$,
all of (a), (b), (c) hold simultaneously with probability $\geq 1-\delta$.
\end{claim}

This means that the number of queries used is
$
\floor*{\frac{n-1}{d}} + 1
= O(1+n/d)
%&= a_(1+n \cdot 32\sigma^2 \ln(6/\delta))\\
= O(1+n\sigma^2\log(1/\delta))$.
\end{proof}
% END \input{static-ub.tex}
% BEGIN \input{static-lb.tex}
\subsection{Lower Bound}
\label{ss:static-lb}

\begin{proof}[Proof of Theorem~\ref{thm:static-lb}]

Since $p>1/n$, it is clear that at least one query is necessary (otherwise one could not do better than randomly guessing the source, which gives $p=1/n$). In the rest of the proof, we show  that one needs $\Omega(\min(p^3n\sigma^2, p^2n))$ queries to identify the source.

 \begin{figure}
\begin{center}
  \includegraphics[width=\textwidth]{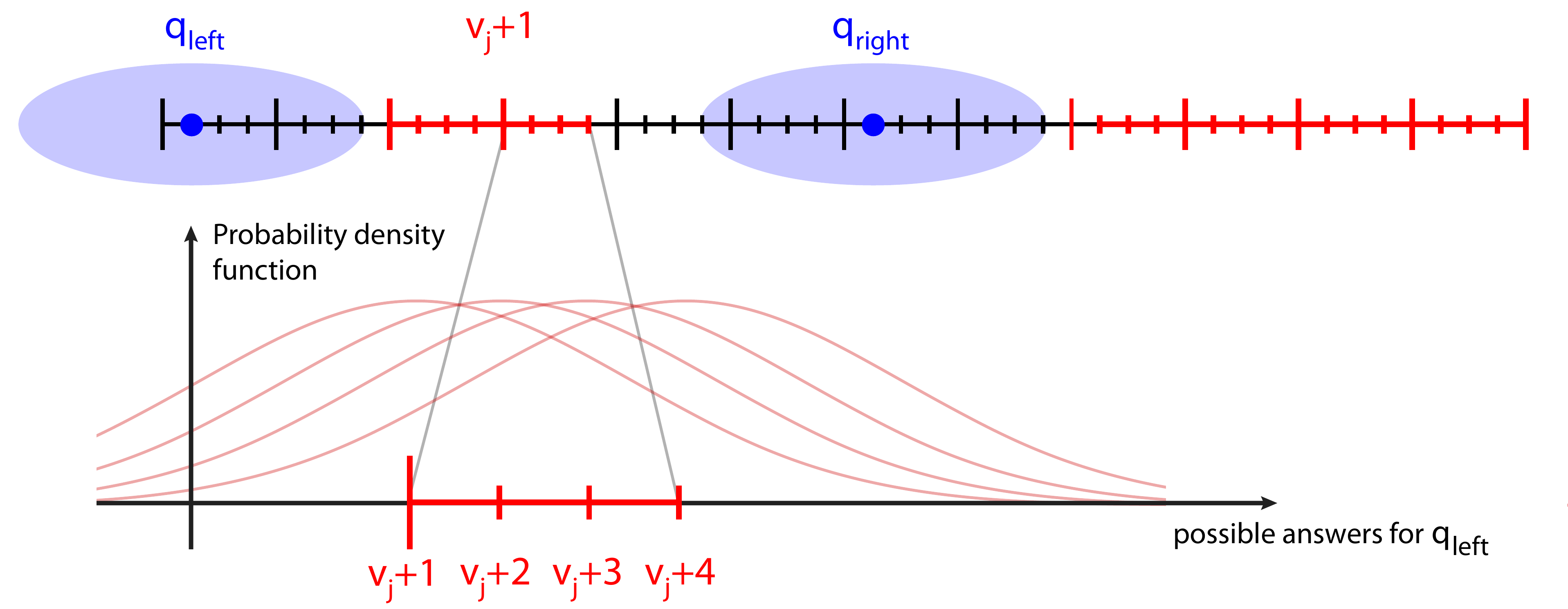}
  \caption{An illustration for the proof of Theorem~\ref{thm:static-lb} with $d=6$ and $k=4$. At the top of the figure, the graph $G$ is shown with the set $Q$ (the query nodes) and the set $C(d)$ (the ``covered'' nodes) marked blue, and the set $K(d,k)$ marked red. At the bottom of the figure, the probability density functions of the answers recorded by $\sLeft$ are shown for each candidate source in the highlighted segment. Intuitively, the union of the areas under the red curves corresponds to the probability of success of the optimal source identification algorithm (this is made concrete in Equation~\eqref{eq:integrals}, although we need to also consider the answer from $\sRight$, so instead of a single integral we get a double integral). In the proof, we show that if too few nodes are queried, then the red segments will be far from the closest query node, and therefore the red curves will have a large overlap, and their union will be small.}\label{fig:static-lb}
  %The circles with blue stroke correspond to the sensors and the red circle signifies the location of the source.}
\end{center}
\end{figure}

Let us introduce some notation.
Let $Q$ be the set of query nodes that the algorithm chooses, let $\dw(v^*,Q) \coloneqq \{\dw(v^*,q)\}_{q \in Q}$ be the answers it receives from each query node, and let $f$ be the function that takes in these answers and returns a prediction for $v^*$.
Since we are not bounding the running time of the algorithm, we can assume that both $Q$ and $f$ are deterministic (the algorithm can simply choose the values of $Q$ and $f$ that give the best chance of finding the source),
while the answers $\dw(v^*,q)$ are random variables depending on both $v^*$ and $w$ (recall that $v^*$ is drawn uniformly in $V=[n]$).
The overall success probability of the algorithm is given by $p = \Pr_{v^*,w}[f(\dw(v^*,Q)) = v^*]$.
For a fixed node $v$, let $p(v) \coloneqq \Pr_w[f(\dw(v,Q)) = v]$ be the probability that the algorithm will output $v$ conditioned on $v^*=v$.
Clearly, $p = \frac{1}{n}\sum_{v \in V} p(v)$. 

As discussed before in Section~\ref{sec:results}, our proofs in the non-adaptive case build on the intuition of query nodes with limited range: roughly speaking, we will show that most nodes are further than a distance $1/\sigma^2$ away from the closest query node, and therefore they will be hard to distinguish from their neighbors. Figure~\ref{fig:static-lb} sketches some of the key points used in the proof.

Let $d>0$ be an integer which we will fix later in equation~\eqref{eq:def-k-and-d}, representing the ``range'' of the query nodes. Intuitively, nodes outside the range~$d$ of any query node (the ``uncovered'' nodes) might be hard to distinguish, contrary to nodes within the range of a query node (the ``covered'' nodes).
Let $C(d) \subset V$ be the set of nodes that are within distance $d$ of some query node in $Q$.

In addition, let us subdivide the first $k\floor*{n/k}$ nodes of $V$ into $\floor*{n/k}$ segments of length $k$, where $k>0$ is an integer that we will fix later, and let $K(d,k) \subseteq V\setminus C(d)$ be the set of nodes contained in the segments that are entirely included in $V\setminus C(d)$ (we will call such segments ``uncovered'').
Our goal in defining these segments is to show that there are few ``covered'' segments, and that the source identification problem is hard to solve on ``uncovered'' segments.%{We will set $k$ to be a bit bigger than $1/p$ in order to show that no algorithm can solve the source location problem on an ``uncovered'' segment of length $k$ with probability better than $p/2$.}

More precisely, to demonstrate that $Q$ needs to be large, we will split the probability of success $p = \frac{1}{n}\sum_{v \in V} p(v)$ into two parts:
\begin{itemize}
\item the part due to $v \in V \setminus K(d,k)$ (the ``covered'' segments), which will be small whenever $Q$, $d$ and $k$ are small (simply because the set $V \setminus K(d,k)$ will be small);
\item the part due to $v \in K(d,k)$ (the ``uncovered'' segments), which will be small ($\leq p/2$) whenever $d$ and $k$ are large enough.
\end{itemize}

Concretely,
\begin{align}
pn
&= \sum_{v \in V} p(v)\notag\\
&= \sum_{v \in V \setminus K(d,k)}p(v) + \sum_{v \in K(d,k)}p(v)\notag\\
&\leq |V \setminus K(d,k)| + \sum_{v \in K(d,k)}p(v).\tag{$p(v)$ is a probability, so $p(v)\leq 1$}\\
&\leq (2d+2k-1)|Q| + (k-1) + \sum_{v \in K(d,k)}p(v)\label{eq:covered-uncovered}
\end{align}
where the factor $(2d+2k-1)$ is because each query node covers $\leq 2d+1$ nodes directly, and can affect $\leq 2(k-1)$ more nodes by touching their segment; also, the $+(k-1)$ comes from the $<k$ nodes that were not within the first $k\floor*{n/k}$ nodes and therefore are not in a segment.

Let us now prove that the sum $\sum_{v \in K(d,k)}p(v)$ is small when $d$ and $k$ are large.
This makes intuitive sense: the nodes in $K(d,k)$ are far from the closest query node, so they will be hard to distinguish from each other. To do this, we will use the following lemma.

\begin{claim}
\label{claim:advantage}
Let $\{v+1, \dots v+k\}$ be a set of $k$ adjacent nodes. Let $\sLeft \in Q$ be the closest query node at or to the left of $v+1$, and let $\sRight \in Q$ be the closest query node at or to the right of $v+k$. Assume that there are no query nodes between $\sLeft$ and $\sRight$, and that $\sLeft \leq v-d$ and $\sRight \geq v+k+d$. Then 
$$\sum_{i=1}^{k} p(v+i) <\frac{2(d+k)e^{\frac{k^2}{2(d+k)\sigma^2}}}{d}.$$ 
%In other words, $f$ cannot distinguish $v^* = v$ from $v^* = v+1$ with more than $O(1/(\sigma\sqrt{d}))$ advantage.
\end{claim}
% BEGIN \input{proofs/tvd.tex}
%\subsection{Proof of Claim~\ref{claim:advantage}}
%\label{proof:claim:advantage}
%Fact~\ref{fact:tvd}}

\begin{proof}%[Proof of Claim~\ref{claim:advantage}]

Let us consider a scenario where the source is sampled uniformly from $\{v+1, \dots, v+k\}$, and let $f' \coloneqq \argmax_f \sum_{i=1}^k \Pr_w[f(\dw(v+i,Q)) = v+i]$ be the algorithm that maximizes the success probability in this scenario. Let $p'(v+i) \coloneqq \Pr_w[f'(\dw(v+i,Q)) = v+i]$, then clearly $\sum_{i=1}^{k} p(v+i)\le\sum_{i=1}^{k} p'(v+i)$ by definition of $f'$.
 Also observe that
$f'$ will only depend on the answers at $\sLeft$ and $\sRight$,
since the algorithm already knows that $v^* \in [\sLeft,\sRight]$, and the other query nodes outside $[\sLeft,\sRight]$ do not carry relevant information. Indeed, any answer outside of $\sLeft$ and $\sRight$ is just the sum of the answer at $\sLeft$ or $\sRight$ plus some extra term that does not say anything about the identity of $v^*$. More formally, one could recreate the other answers from just the answers at $\sLeft$ and $\sRight$ in a way that exactly replicates the original distribution, so we can transform any algorithm that uses all the answers into an algorithm that uses only the answers at $\sLeft$ and $\sRight$ with the same performance.

For similar reasons, we can assume that $\sLeft=v-d$ and $\sRight=v+k+d$. If the query nodes were any further, that would be more difficult for the algorithm $f'$, because we can simulate that case using the answers at $v-d$ and $v+k+d$.

\newcommand{\Dl}{{a_{\text{left}}}}
\newcommand{\Dr}{{a_{\text{right}}}}

%{$\W_{v+i}=(\Normal(d+i, (d+i)\sigma^2), \Normal(d+k-i, (d+k-i)\sigma^2)$ for $i \in \{1,\ldots,k\}$ (where each $\W_{v+1}$ is a couple of \emph{independent} normal distributions)}
Since $f'$ maximizes the success probability of estimating the hidden parameter $v^*$ with both the likelihood function and the prior distribution being completely known, the optimal $f'$ computes the posterior distribution using Bayes rule, and picks the $v^*$ that maximizes it (this is called Maximum A Posteriori or MAP estimation) \cite{PintoTV12}. In our case we have a uniform prior, which implies that $f'$ is simply the Maximum Likelihood Estimator, i.e., for any answer $(\Dl,\Dr)$,
$$f'(\Dl,\Dr) = \argmax\limits_{i \in \{1,\ldots,k\}} (g_{v+i}(\Dl,\Dr)),$$
where $g_{v+i}(x,y)$ denotes the probability density function of $\W_{v+i}=(\Normal(d+i, (d+i)\sigma^2), \Normal(d+k-i, (d+k-i)\sigma^2)$, the distribution of the answers at $\sLeft$ and $\sRight$ (note that $\W_{v+i}$ is a pair of \emph{independent} normal distributions). Consequently,
\begin{equation}
\label{eq:integrals}
\sum_{i=1}^{k} p'(v+i) = \int\limits_{-\infty}^{\infty}  \int\limits_{-\infty}^{\infty}   \max\limits_{i \in \{1,\ldots, k\}} (g_{v+i}(x,y)) \diff x \diff y.
\end{equation}
Next, we provide the following upper bound to $g_{v+i}$ for every $i \in \{1, \ldots, k\}$:
\begin{align}
\label{eq:scaled_pdf1}
g_{v+i}(x,y) &=\frac{1}{2\pi \sigma^2\sqrt{(d+i)(d+k-i)}}\mathrm{exp}\left(-\frac{(x-(d+i))^2}{2(d+i)\sigma^2} -\frac{(y-(d+k-i))^2}{2(d+k-i)\sigma^2}  \right) \nonumber\\
&< \frac{1}{2\pi d\sigma^2} \mathrm{exp}\left(-\frac{ (x-(d+i))^2 + (y-(d+k-i))^2  }{2(d+k)\sigma^2}  \right)
\end{align}

Notice that if we consider a triangle $ABC$ with $A=(x,y)$, $B=(d+i, d+k-i)$ and $C=(d,d)$, and we denote the side lengths opposite of each point by $a,b$ and $c$, then the numerator of the exponent in equation \eqref{eq:scaled_pdf1} equals $c^2$. The following lower bound holds for $c^2$ based on the law of cosines and elementary algebra:
$$c^2 = a^2 -2ab\cos(\measuredangle ACB) + b^2 \ge a^2 -2ab + b^2 \ge \frac{a^2}{2}-b^2.$$
The last inequality can be confirmed if we move all terms to the left side and find the expression $(a/\sqrt{2}-\sqrt{2}b)^2 \ge 0$. After substituting back into $a,b$ and $c$, since the maximum distance between points $(d+i, d+k-i)$ and $(d,d)$ is $k$ for any $i \in \{ 1, \dots, k \}$, we get
\begin{equation}
\label{eq:abc_consequence}
(x-(d+i))^2 + (y-(d+k-i))^2 \ge\frac12 ((x-d)^2 + (y-d)^2) - k^2.
\end{equation}
Substituting equation \eqref{eq:abc_consequence} back into equation \eqref{eq:scaled_pdf1} yields
\begin{align}
g_{v+i}(x,y) &< \frac{1}{2\pi d\sigma^2} \mathrm{exp}\left(-\frac{ (x-d)^2 + (y-d)^2 -2k^2 }{4(d+k)\sigma^2} \right) \nonumber \\
\label{eq:scaled_pdf2}
&=\frac{2(d+k)e^{\frac{k^2}{(d+k)2\sigma^2}}}{d} \cdot \frac{1}{2\pi (d+k)2\sigma^2} \mathrm{exp}\left(-\frac{ (x-d)^2 + (y-d)^2}{2(d+k)2\sigma^2} \right).
\end{align}
Notice that the last line of equation \eqref{eq:scaled_pdf2} can be written as 
$$\frac{2(d+k)e^{\frac{k^2}{2(d+k)\sigma^2}}}{d}g(x,y),$$
where $g(x,y)$ is the probability density function of two independent copies of $\Normal(d,(d+k)2\sigma^2)$, so its double integral must sum to $1$. Thus, plugging this upper bound into equation \eqref{eq:integrals}, we get
$$\sum_{i=1}^{k} p'(v+i) < \frac{2(d+k)e^{\frac{k^2}{2(d+k)\sigma^2}}}{d} \int\limits_{-\infty}^{\infty} \int\limits_{-\infty}^{\infty} g(x,y) \diff x \diff y= \frac{2(d+k)e^{\frac{k^2}{2(d+k)\sigma^2}}}{d}.$$
Since $\sum_{i=1}^{k} p'(v+i)$ is an upper bound on $\sum_{i=1}^{k} p(v+i)$, the proof is completed.
\end{proof}
% END \input{proofs/tvd.tex}

%%%%%%%%%%%%%%%%%%

Since each segment in $K(d,k)$ contains $k$ consecutive nodes that are all a distance $d$ away from the closest query node, we can apply Claim~\ref{claim:advantage} to each of them, and get
\begin{align*}
    \sum_{v \in K(d,k)}p(v)
    &\leq \frac{|K(d,k)|}{k}\cdot \frac{2(d+k)e^{\frac{k^2}{2(d+k)\sigma^2}}}{d}\\
    &\leq \frac{n}{k}\cdot \frac{2(d+k)e^{\frac{k^2}{2(d+k)\sigma^2}}}{d}.
\end{align*}
In order to make this $\leq pn/2$, let us set
\begin{equation}
\label{eq:def-k-and-d}
k\coloneqq \ceil*{16e/p}\qquad\text{and}\qquad d\coloneqq \max\left(\ceil*{\frac{k^2}{2\sigma^2\ln(pk/8)}}, k\right)
\end{equation}
(this value of $k$ is chosen so that the $\ln(\cdot)$ in the definition of $d$ is positive\footnote{This value of $k$ also makes sense intuitively: we are showing that no algorithm can solve the source identification problem with probability better than $p/2$ in an uncovered interval, so we definitely need at least $k\geq 2/p$ to rule out random guessing.}).
Indeed, these choices give
\begin{align}
    \sum_{v \in K(d,k)}p(v)
    &\leq \frac{n}{k}\cdot \frac{2(d+k)e^{\frac{k^2}{2(d+k)\sigma^2}}}{d}\notag\\
    &\leq \frac{n}{k}\cdot \frac{4d}{d}e^{\frac{k^2}{2d\sigma^2}}\tag{$k\leq d$ by \eqref{eq:def-k-and-d}}\\
    &\leq \frac{n}{k}\cdot 4e^{\ln(pk/8)}\tag{$d\geq \frac{k^2}{2\sigma^2\ln(pk/8)}$ by \eqref{eq:def-k-and-d}}\\
    &= pn/2.\label{eq:uncovered-pnhalf}
\end{align}
Combining \eqref{eq:covered-uncovered} with \eqref{eq:uncovered-pnhalf}, we finally get
\begin{align*}
    pn
    &\leq (2d+2k-1)|Q| + (k-1) + pn/2
\end{align*}
which further implies
\begin{align*}
 |Q|&\geq \frac{pn/2 - k}{2d+2k}\\
    &\geq \frac{pn - 2k}{8d}.\tag{$k \leq d$}
\end{align*}
Let us assume $k \leq pn/4$ (otherwise we have $p^2n < 4pk \leq 4p \ceil*{\frac{16e}{p}} = O(1)$, which means the $\Omega(\min(p^3n\sigma^2, p^2n))$ bound we are trying to prove becomes a trivial $\Omega(1)$).
Then we get
\begin{align*}
    |Q| &\geq \frac{pn-2k}{8d}\\
    &\geq \frac{pn}{16d}\tag{$k \leq pn/4$}\\
    &\stackrel{\eqref{eq:def-k-and-d}}{=} \frac{pn}{16\max\left(\ceil*{\frac{k^2}{2\sigma^2\ln(pk/8)}}, k\right)}\tag{by definition of $d$}\\
    &\stackrel{\eqref{eq:def-k-and-d}}{=} \frac{pn}{O\left(\max\left(\frac{1}{p^2\sigma^2}, \frac{1}{p}\right)\right)}\tag{by definition of $k$}\\
    &= \Omega(\min(p^3n\sigma^2, p^2n)).\qedhere
\end{align*}

\end{proof}
% END \input{static-lb.tex}
% END \input{static.tex}
% BEGIN \input{sequential.tex}
\section{Proofs for the Adaptive Setting}
%In this section, we prove Theorems~\ref{thm:sequential-ub} and~\ref{thm:sequential-lb}. This shows that the query complexity in the sequential setting is $\Theta(1 + \log (1+\log_{1/\sigma} n))$ when the variance is low ($\sigma^2 \leq 1/2$), and $\Theta(\log \log n) + O_\sigma(1)$ when the variance is high ($\sigma^2 \geq 1/2$).
% BEGIN \input{sequential-ub.tex}
\subsection{Upper Bound}

\begin{proof}[Proof of Theorem~\ref{thm:sequential-ub}]
The algorithm crucially uses the following result on the concentration of the answers at large distances. We prove it in~\ref{proof:lemma:concentration-ub}.
\begin{lemma}
\label{lemma:concentration-ub}
For any probability $0 < \delta < 1/2$, there is some constant $C(\delta)=O\left(\sqrt{\log(1/\delta)}\right)$ such that for any $n,\sigma$ and any source $v^* \in V$, we have
\begin{equation}\label{eq:concentration-ub}
\Pr_w[\forall q \in V, \,\dw(v^*,q) \in |v^*-q| \pm C(\delta)\cdot\sigma\sqrt{|v^*-q|}\ln (1+|v^*-q|)] \geq 1-\delta.
\end{equation}
That is, the concentration bound $|v^*-q| \pm C(\delta)\cdot\sigma\sqrt{|v^*-q|}\ln (1+|v^*-q|)$ holds simultaneously for all nodes $q$ with probability at least $1-\delta$ over the choice of the weights $w$.
\end{lemma}

With this concentration result in hand, the algorithm follows a natural ``iterative refining'' strategy: start by obtaining a rough estimate of the identity of $v^*$, then progressively refine it by querying nodes closer and closer to $v^*$. After $k$ steps (where $k$ is defined by Claim~\ref{claim:asymptotics-ub}), only a few possible candidate sources will remain, and the algorithm will switch to testing them one by one.

Concretely, let us assume that Lemma~\ref{lemma:concentration-ub} holds with the desired probability of failure $\delta$. Then the algorithm will maintain a shrinking interval $[l_i,r_i]$ which contains $v^*$. Initially $l_0=1$ and $r_0=n$. At each step, the algorithm will query node $l_i$.
Let $d_i$ be equal to $r_i-l_i$. $v^*$ has to be to the right of $l_i$ and at distance at most $d_i$ from $l_i$, so

\begin{align*}
\dw(v^*,l_i) &\in |v^*-l_i| \pm C(\delta)\cdot \sigma\sqrt{|v^*-l_i|}\ln(1+|v^*-l_i|),\tag{by \eqref{eq:concentration-ub}}\\
&\subseteq v^*-l_i \pm C(\delta)\cdot\sigma\sqrt{d_i}\ln(1+d_i)\tag{$v^* \geq l_i$ and $|v^*-l_i| \leq d_i$}
\end{align*}
and thus given answer $\dw(v^*, l_i)$ the algorithm knows that $v^*$ must be in interval
%\[[l_i+\dw(v*,l_i)-C(\delta)\cdot\sigma\sqrt{d_i}\ln(1+d_i),\,l_i+\dw(v^*,l_i)+C(\delta)\cdot\sigma\sqrt{d_i}\ln(1+d_i)]\]
\begin{equation}
\label{eq:seq_upper_interval}
l_i+\dw(v^*,l_i) \pm C(\delta)\cdot\sigma\sqrt{d_i}\ln(1+d_i).
\end{equation}
Therefore, it shrinks its interval as follows:
\[
\begin{cases}
l_{i+1} = \max(l_i, l_i+\ceil{\dw(v^*,l_i) - C(\delta)\cdot\sigma\sqrt{d_i}\ln(1+d_i)})\\
r_{i+1} = \min(r_i, l_i+\floor{\dw(v^*,l_i) + C(\delta)\cdot\sigma\sqrt{d_i}\ln(1+d_i)})
\end{cases}
\]
The resulting interval has length $d_{i+1} = r_{i+1}-l_{i+1} \leq 2\,C(\delta)\cdot\sigma\sqrt{d_i}\ln(1+d_i)$.

Now what remains to do is to figure out how fast this interval shrinks, and when we should switch to testing the remaining candidates one by one. To get a rough initial intuition of the speed at which it shrinks, let us imagine that $d_{i+1} = \sigma\sqrt{d_i}$. Then, the sequence would decrease very fast at the onset, when $d_i$ is still large, then decrease slower and slower. We would observe that $\log(d_{i+1}/\sigma^2) = \log(\sqrt{d_i}/\sigma) = \frac{1}{2}\log(d_i/\sigma^2)$: the \emph{logarithm} of the ratio of $d_i$ to $\sigma^2$ is divided by 2 at each step. So it would be reasonable to assume that $d_i$ will approach $\sigma^2$ in a doubly-logarithmic number of steps. This is made rigorous in the following claim, which is proved in~\ref{proof:claim:asymptotics-ub}.

\begin{claim}
\label{claim:asymptotics-ub}
Assume $d_0\leq n$, and $d_{i+1} \leq C \cdot\sigma\sqrt{d_i}\ln(1+d_i)$ for some value $C>0$. Then
\begin{itemize}
\item if $\sigma^2 \leq 1/2$, there exists $k = O(\log \log_{1/\sigma} n)$ such that $d_k = \poly(C)$;
\item if $\sigma^2 \geq 1/2$, there exists $k = O(\log \log n)$ such that $d_k = \sigma^2\cdot \poly(C, \log(1+\sigma^2))$.
\end{itemize}
\end{claim}

We instantiate Claim~\ref{claim:asymptotics-ub} with $C \coloneqq 2\,C(\delta)$. After the first $k$ steps, we simply go through the $d_k+1$ remaining possible positions for node $v^*$ in $[l_k,r_k]$, and check them all with one query each.\footnote{As we explain in~\ref{sec:other-distributions}, this theorem can be extended to apply to many other edge weight distributions. If the distribution's support is positive, a binary search can be used instead.} With probability $1$, $v^*$ will be the only one to give $0$ as an answer . Thus, overall, this algorithm will succeed with probability at least $1-\delta$.
The total number of queries used is $k+d_k+1$, which by Claim~\ref{claim:asymptotics-ub} gives the desired bounds in both the $\sigma^2 \leq 1/2$ case and the $\sigma^2 \geq 1/2$ case.
\end{proof}

% END \input{sequential-ub.tex}
% BEGIN \input{sequential-lb.tex}
\subsection{Lower Bound}
\label{sec:sequential-lb}
% BEGIN \input{defs-sequential-lb.tex}
\newcommand{\Typical}{\mathrm{Typical}}
\newcommand{\oL}[1]{a_{\mathrm{L},#1}}
\newcommand{\oR}[1]{a_{\mathrm{R},#1}}
\newcommand{\reduce}{\mathrm{reduce}_{n,\sigma}}
\newcommand{\jMin}{{j_{\mathrm{min}}}}
\newcommand{\jStop}{{j_{\mathrm{stop}}}}
\newcommand{\SenRV}{\{q_i\}_{\leq j}}
\newcommand{\ObsRV}{\{a_i\}_{\leq j}}
\newcommand{\dilution}[1]{\frac{#1}{\left(\frac{8}{3}\lambda_{j+1}+1\right)\log n}}
\newcommand{\GroupComplex}{E_{q_{l_j},q_{r_j},a_{l_j},a_{r_j}}}
\newcommand{\Group}{E_{l,r,a_l,a_r}}
\newcommand{\argmin}{\mathrm{argmin}}
\newcommand{\wIn}{w_{\mathrm{in}}}
\newcommand{\wOut}{w_{\mathrm{out}}}
\newcommand{\mup}{\mu'}
\newcommand{\cmup}{\mup}
\newcommand{\ssmu}{\sigma\sqrt{\mu}}
\newcommand{\myDl}{d_l'(v')}
\newcommand{\myDr}{d_r'(v')}
\newcommand{\vpeak}{v'_{\mathrm{peak}}}
\newcommand{\lMin}{{\lambda_\text{min}}}
% END \input{defs-sequential-lb.tex}

For your reading convenience, here is a quick reference of the notations that are used throughout the proof.

\noindent\fbox{%
    \parbox{\textwidth}{%
        \emph{Notations cheatsheet (not exhaustive)}
        \begin{itemize}
            \item $p$: the desired probability of identifying the true source.
            \item $R$: the internal randomness of the algorithm (see Definition~\ref{def:r}).
            \item $q_j$: the  $j\th$ node queried by the algorithm.
            \item $a_j$: the $j\th$ answer the algorithm receives (it will take the value $\dw(v^*, q_j)$).
            \item $T$: shorthand for $\Typical_{p/2}(v^*, w)$, i.e. the event that the concentration bounds from Definition~\ref{def:typical} hold.
            \item $C$: shorthand for $C(p/2)$, a constant (in $n$ and $\sigma$) factor involved in the concentration bounds of Definition~\ref{def:typical} (see Definition~\ref{def:cd} for its precise value).
            \item $D$: shorthand for $D(\sigma, p/2)$, the minimum distance at which the concentration bounds of Definition~\ref{def:typical} hold (see Definition~\ref{def:cd} for its precise value).
            \item $l_j, r_j$: the step counters at which the closest query nodes to $v^*$ have been placed so far at step $j$; i.e. by the time the $j\th$ query has been asked, $v^*$ lies between query nodes $q_{l_j}$ and $q_{r_j}$.
            \item $\mu_j$: the minimum of the answers to query nodes $q_{l_j}$ and $q_{r_j}$. Intuitively, it is a proxy for the smallest answer made so far (after asking $j$ queries).
            \item $\reduce(x)$: a function $\R \to \R$ that models the fastest decrease of $\mu_j$ an algorithm can hope for: most of the time $\mu_{j+1} \geq \reduce(\mu_j)$ (see Definition~\ref{def:reduce}).
            \item $\lambda_i$: a lower bound on $\mu_j$ with high probability, and therefore a limit on the progress that the algorithm can make (see Definition~\ref{def:lambda}).
            \item $K_j$: a random variable representing all the information that the algorithm has at its disposal after asking the first $j$ queries (see Definition~\ref{def:k}).
            \item $A_j$: the event that $\mu_j \geq \lambda_j$; intuitively, the event that at step $j$, the algorithm has not queried any nodes very close to $v^*$.
            \item $B_j$: informally, the event that even based on everything the algorithm knows at step $j$, no node is particularly likely to be the source (see Definition~\ref{def:b}).
            \item $\jStop$ (also, $\jMin$): a lower bound on the number of steps that an algorithm needs to find the source with probability $p$ (see Definition~\ref{def:jstop}).
        \end{itemize}
    }%
}

\begin{definition}[$R$]
\label{def:r}
Let $R$ be a random variable denoting the internal randomness of the algorithm. One can for example think of $R$ as drawn uniformly from interval $[0,1]$, as this puts no limitation on the amount of randomness the algorithm can use.
\end{definition}

\begin{definition}[$q_j, a_j$]
\label{def:so}
Let $q_j$ be the $j\th$ query node selected by the algorithm, and let $a_j$ be the answer that it gets to query $q_j$ (i.e. $a_j\coloneqq\dw(v^*, q_j)$). Both $q_j$ and $a_j$ are random variables that can depend on $v^*$, $w$ and the internal randomness of the algorithm.
\end{definition}

Note that under the hypotheses of Theorem~\ref{thm:sequential-lb}, any algorithm needs to query at least one node: otherwise it would succeed with probability at most $1/n < p$. 

To simplify the proof, we make the following adaptations to the model, which only give the algorithm more power to identify the source, and therefore hold without loss of generality:
\begin{alphaenumerate}
\item Before the algorithm starts, two initial query nodes $q_{-1}=1$ and $q_0=n$ are already selected, resulting in answers $a_{-1}$ and $a_0$ at no cost to the algorithm. The first query that is actually chosen by the algorithm is $q_1$.
\item When querying node $q_j$, in addition to the answer $a_j $, the algorithm is told on which side of $q_j$ the source $v^*$ is identified.\footnote{Note that this gives the algorithm the ability to perform a binary search, which is not necessarily easy when the weight distribution is not positive.}
% This means that the algorithm always knows between which two of the sensors placed so far $v^*$ lies.
\item Once the algorithm is ready to guess the position of $v^*$, it should query it.\footnote{This extra query does not affect the asymptotics because as noted in the previous paragraph the algorithm always needs to query at least one node anyway.} If at any point the algorithm queries $v^*$, it immediately terminates and the identification is considered successful. More precisely, the number of queries that the algorithm uses is defined as the first positive integer $j$ such that $q_j=v^*$.

\end{alphaenumerate}

The details of the proof are at times technically heavy, so we first give a general outline to provide the gist of the proof. It proceeds in the following 8 steps. We will cover each of them in detail in the next 8 subsections (Section~\ref{sec:sequential-lb}.$x$ corresponds to step $x$).
\begin{enumerate}
\item We define a random event $T$ (over $v^*$ and $w$) which has probability $\geq 1-p/2$, guarantees that $v^*$ is not too close to the ends of the path, and gives some concentration bounds for the answers when the query node is at least some distance $D$ away from $v^*$ ($D$ will be defined in Definition~\ref{def:t}). $T$ represents a ``typical situation'': the role of this event is to exclude some extreme cases (e.g. $v^*=1$ or $v^*=n$) that would derail the proof.
\item We define a sequence of random variables $\mu_j$ that describe how close the algorithm is to find $v^*$ after it has asked $j$ queries. $\mu_j$ is (roughly) the distance between $v^*$ and the query node closest to $v^*$, and tends to decrease as $j$ increases. We also define a corresponding \emph{deterministic} sequence $\lambda_0 > \lambda_1 > \cdots$ where for each $j$, $\mu_j \geq \lambda_j$ with high probability.
\item We define the following events which (basically) imply each other in alternation (i.e. $A_j \Rightarrow B_j \Rightarrow A_{j+1}$), and will help us bound the progress of the algorithm:
    \begin{itemize}
    \item $A_j$ is the event that $\mu_j \geq \lambda_j$; it intuitively means ``none of the $j$ first query nodes are too close to $v^*$'';
    \item $B_j$ will be defined later, and intuitively means ``even after asking $j$ queries, the algorithm has only a vague idea where $v^*$ is'', or a bit more precisely, ``even conditioned on all the answers gathered by the algorithm during the first $j$ steps, none of the nodes have a high probability of being the source''.
    \end{itemize}
\item We define $\jStop$ be the largest $j$ such that $\lambda_j \geq D$ (recall that $D$ is the distance above which event $T$ gives concentration bounds on answers). Our goal will be to prove that with high probability, the algorithm needs to ask at least $\jStop$ queries.
\item We prove two key lemmas, which show that in most cases, $A_j \Rightarrow B_j$ and $B_j \Rightarrow A_{j+1}$. They state that for $j< \jStop$,
    \begin{itemize}
    \item $A_j$ implies $B_j$ (Lemma~\ref{lemma:a-implies-b});
    \item with probability $1 - \frac{1}{\log n}$, $T \wedge A_j \wedge B_j$ implies $A_{j+1}$ (Lemma~\ref{lemma:b-implies-a}).
    \end{itemize}
    This is the core technical part of the proof.
\item We chain the above lemmas by induction and use the fact that $\Pr[\neg T] \leq p/2$ to obtain $\Pr[A_{\jStop}] \geq 1-p$.
\item We prove that $\jStop = \Omega(\log \log_{\max(1/\sigma, 2)} n)$, the desired lower bound.
\item We observe that event $A_{\jStop}$ implies that the algorithm has not found $v^*$ after asking $\jStop$ queries, which using 6 and 7 completes the proof.
\end{enumerate}

\subsubsection{Typical instances: event $T$}
In our model, there are no hard guarantees on how far away the answer $\dw(v^*,q)$ might be from the real distance $|v^*-q|$. For example, $\dw(v^*,v^*+1) \sim \Normal(1,\sigma^2)$ might be as large as $1000$, even if $\sigma=1$ (though with very low probability). While such extreme events are intuitively disadvantageous for the algorithm, they also make it harder to prove lower bounds. Therefore, we need to make basic assumptions on the range of $\dw(v^*,q)$ at high distances.

To do this, we will need to use the notion of a ``typical'' instance: a choice of $v^*$ and $w$ for which some reasonable concentration results hold. Note that part (i) below is very similar to Lemma~\ref{lemma:concentration-ub}, which we used for the upper bound.
Let $C(\delta)$ and $D(\sigma,\delta)$ be two values (defined later in Definition~\ref{def:cd}) such that
\begin{equation}
\label{eq:guarantees-on-cd}
\max(\sigma^2, e^2) \leq D(\sigma,\delta) = o_\delta(\max(\sigma^2\log\sigma, 1)).
\end{equation}

\begin{definition}[$\Typical_\delta(v^*,w)$]
\label{def:typical}
For any probability $\delta>0$, let $\Typical_\delta(v^*,w)$ be the event that the following holds:% (for $C(\delta)$ and $D(\sigma,\delta)$ to be defined later):
\begin{alphaenumerate}
\item $\min(\dw(v^*,1),\dw(v^*,n)) \geq \frac{n}{C(\delta)}$;
\item for all $q$ with $d_{q} \coloneqq |v^*-q| \geq D(\sigma,\delta)$,
\begin{romanenumerate}
\item $\dw(v^*,q) \in d_{q} \pm \sigma \sqrt{d_{q}}\ln d_{q}$
\item $\dw(v^*,q) \in d_{q} \pm d_{q}/4 = [\frac{3}{4}d_{q}, \frac{5}{4}d_{q}]$.
\end{romanenumerate}
\end{alphaenumerate}
Part (a) means that the two answers from the query nodes at either end of the path are not too much smaller than their expectation $\Omega(n)$, and part (b) means that above a certain distance threshold $D(\sigma, \delta)$, all answers are concentrated around their mean.
\end{definition}

%\begin{fact}
%\label{fact:d}

%\end{fact}

As the name indicates, most instances are typical (the proof %of Lemma~\ref{lemma:typical-is-likely}
is given in~\ref{proof:lemma:typical-is-likely}).
\begin{lemma}
\label{lemma:typical-is-likely}
For any probability $\delta>0$ and any $n \geq \Theta_\delta(\max(\sigma^2\ln\sigma, 1))$,
$\Pr_{v^*,w}[\Typical_{\delta}(v^*,w)] \geq 1-\delta$.
\end{lemma}
We will apply Lemma~\ref{lemma:typical-is-likely} with $\delta\coloneqq p/2$. We will use the following shorthands.
\begin{definition}[$T,C,D$]
\label{def:t}
Let $T \coloneqq \Typical_{p/2}(v^*,w)$, $C \coloneqq C(p/2)$ and $D\coloneqq D(\sigma,p/2)$.
\end{definition}
%We immediately obtain:
\begin{corollary}
\label{corollary:typical-is-likely}
$\Pr[T] \geq 1-p/2$.
\end{corollary}

\subsubsection{Measure of progress $\mu_j$ and benchmark $\lambda_j$}
\label{sss:mu}
It turns out that the right metric of progress to look at is (roughly speaking) the \emph{smallest answer value seen so far}. More precisely, suppose that the algorithm has asked $j$ queries so far (and hence is at step $j$). Then we define the quantity $\mu_j$ as follows.
\begin{definition}[$l_j,r_j,\mu_j${}]
\label{def:mu}
Let $l_j \coloneqq \argmax_{i\leq j, q_i\leq v^*}(q_i)$ and $r_j \coloneqq \argmin_{i \leq j,q_i \geq v^*}(q_i)$, which means that $q_{l_j}$ (resp. $q_{r_j}$) is the closest query node at or to the left (resp. right) of $v^*$ placed so far.
%Let $q_j=\{q_{-1},q_0,q_1,\ldots,q_j\}$ be the set of sensors placed so far.\footnote{Recall that sensors $q_{-1}=1$ and $q_0=n$ are given away for free.} Let $l_j=\max(q_j \cap [1,v^*])$ and $r_j=\max(q_j \cap [v^*,n])$ be the closest sensor to the left and right of $v^*$ so far (or $v^*$, if it has been found).
%Let $\oL{j}\coloneqq \dw(v^*,l_j)$ and $\oR{j} \coloneqq \dw(v^*,r_j)$ be the corresponding observations.
Then $\mu_j \coloneqq \min(a_{l_j},a_{r_j})$, the smaller of the corresponding answers.
%Then $\mu_j \coloneqq \min(\dw(v^*,l_j),\dw(v^*,r_j))$.
\end{definition}
Note in passing that by simplifying assumption (b) in the beginning of Section~\ref{sec:sequential-lb}, the algorithm knows $l_j$ and $r_j$.
%, and therefore it also knows $\mu_j$.
Also, if $\mu_j > 0$, then the algorithm has not found $v^*$ yet (otherwise we would have $q_{l_j} = q_{r_j} = v^*$ and thus $a_{l_j} = a_{r_j} = 0$).
%\begin{fact}
%$T$ implies $\mu_0 \geq n/C$.
%\end{fact}
%\begin{fact}
%If 
%\end{fact}

We want to show that, with high probability, $\mu_j$ cannot decrease too fast with $j$. To make it formal, we define an analogous \emph{deterministic} sequence $\lambda_j$, which we will show is a lower bound for $\mu_j$ with high probability. We call $\lambda_j$ a ``benchmark'' because it is a point of comparison to determine whether the algorithm is making fast progress or not. It decreases with $j$ according to the following function.
\begin{definition}
\label{def:reduce}
Let $\reduce(x) = \frac{\sigma\sqrt{x}}{400\ln x\log n}$.
\end{definition}
%This allows us to define the benchmark sequence $\lambda_j$.
\begin{definition}
\label{def:lambda}
Let $\lambda_0 \coloneqq n/C$ and $\lambda_{j+1} \coloneqq \reduce(\lambda_j)$.\footnote{If $\lambda_j\leq 0$, we define $\lambda_{j+1}\coloneqq 0$. However, we will never use such values.}
\end{definition}
Observe that by point (a) in Definition~\ref{def:typical}, $T$ implies $\mu_0 \geq n/C = \lambda_0$. Our goal will be to prove that $\mu_j \geq \lambda_j$ will likely continue to hold as $j$ increases.

\subsubsection{Events $A_j$ and $B_j$}
Informally, at step $j$, $A_j$ is the event that the algorithm has not queried any nodes very close to $v^*$, and $B_j$ is the event that the algorithm has only a vague idea of where $v^*$ is (or more precisely, that even conditioned on all the answers so far, no node has a high probability of being the source). As we will see in Section~\ref{sec:key-lemmas}), intuitively,
\begin{itemize}
\item $A_j$ implies $B_j$ because if the algorithm does not have any query nodes close to $v^*$, then the answers it got are all very noisy, and thus its confidence interval for $v^*$ is wide (Lemma~\ref{lemma:a-implies-b});
\item $B_j$ implies $A_{j+1}$ because if all nodes are very unlikely to be the source $v^*$, then wherever it decides to query the next node, it is unlikely to be very close to $v^*$ (Lemma~\ref{lemma:b-implies-a}).
\end{itemize}

As we will see, both events depend only on information that is available to the algorithm at step $j$. For convenience, we define random variable $K_j$, which describes all the knowledge of the algorithm up to step $j$. %We will write the events as $A_j(K_j)$, $B_j(K_j)$ when there is a need to disambiguate.
\begin{definition}[$K_j$]
\label{def:k}
Let $\SenRV \coloneqq (q_{-1}, \ldots, q_j)$ and $\\ \coloneqq (a_{-1}, \ldots, a_j)$ be the query nodes and answers available at step $j$. Then let $K_j = (\SenRV, \ObsRV, l_j, r_j)$. This encodes the locations of all query nodes, the answers received from them, as well as the identity of the two query nodes between which $v^*$ lies.
\end{definition}
$A_j$ is the event that $\mu_j$ is greater than the benchmark $\lambda_j$.
\begin{definition}
\label{def:a}
Let $A_j$ be the event that $\mu_j \geq \lambda_j$.
\end{definition}
$B_j$ is the event that the posterior of $v^*=v$ given $K_j$ is ``diluted''.
\begin{definition}
\label{def:b}
Let $B_j$ be the event that for all nodes $v \in V$,
\[\Pr[T \wedge (v^*=v) \mid K_j] \leq \dilution{1}.\]
\end{definition}
Note that $\Pr[T \wedge (v^*=v) \mid K_j]$ itself is a random variable since it depends on $K_j$, so $B_j$ is still a random event even though it is a statement about a probability. An equivalent way to define $B_j$ is to first define random variable
\[P_j \coloneqq \max_{v \in V} \Pr[T \wedge (v^*=v) \mid K_j],\]
then to let $B_j$ be the event that $P_j \leq \dilution{1}$.

\subsubsection{Stopping step $\jStop$}
Our goal is to show that for a high value of $j$, we have $\mu_j > 0$ with high probability, and therefore the algorithm has failed to find $v^*$ using only $j$ queries. We now define that value of $j$.
\begin{definition}[$\jMin,\jStop$]
\label{def:jstop}
Let $\jMin$ be the smallest integer $j \geq 0$ such that $\lambda_j < D$. Then
\[\jStop \coloneqq \min\left(\jMin-1, \floor*{\frac{p\log n}{2}}\right).\]
\end{definition}
This means that at step $j \leq \jStop$, $\lambda_j \geq D$ is still big enough for the concentration bounds of event $T$ to hold. The second argument of the $\min(\cdot,\cdot)$ is just for convenience of the proof, and will not matter if $n$ is large enough. We will also use the following easily believable fact, proved in~\ref{proof:fact:lambda-decreasing}.
\begin{fact}
\label{fact:lambda-decreasing}
$\lambda_0 > \lambda_1 > \cdots > \lambda_{\jStop} > \lambda_{\jStop+1}$.
\end{fact}

\subsubsection{Key lemmas}
\label{sec:key-lemmas}
We now state our two main lemmas, which constitute the core technical part of the proof. The proof of Lemma~\ref{lemma:a-implies-b} is very technical and not particularly enlightening, so it is deferred to~\ref{proof:lemma:a-implies-b}. The proof of Lemma~\ref{lemma:b-implies-a}, on the other hand, is much more straightforward, and we include it here.
\begin{lemma}
\label{lemma:a-implies-b}
If $j < \jStop$, then $A_j \Rightarrow B_j$.
\end{lemma}
\begin{lemma}
\label{lemma:b-implies-a}
If $j \leq \jStop$, then $\Pr[\neg T \vee \neg A_j \vee \neg B_j \vee A_{j+1}] \geq 1-\frac{1}{\log n}$.
\end{lemma}
Note that ``$\neg T \vee \neg A_j \vee \neg B_j \vee A_{j+1}$'' is logically equivalent to ``$(T \wedge A_j \wedge B_j) \Rightarrow A_{j+1}$''. 
Intuitively, if $B_j$ holds, then the probability of $v^* = v$ (conditioned on the answers so far) is low for any $v$, which means that whatever the algorithm picks as its next query $q_{j+1}$, the probability that $q_{j+1}$ is within some distance $d$ of $v^*$ is upper bounded by the sum of those probabilities over $v \in [q_{j+1}-d,q_{j+1}+d]$. Therefore, with high probability, $a_{j+1}$ will not be too small, and the same holds for $\mu_{j+1}$.

% BEGIN \input{proofs/b-implies-a.tex}
%\subsection{Proof of Lemma~\ref{lemma:b-implies-a}}
%\label{proof:lemma:b-implies-a}

\begin{proof}[Proof of Lemma~\ref{lemma:b-implies-a}]
We will show equivalently that $\Pr[T \wedge A_j \wedge B_j \wedge \neg A_{j+1}] \leq \frac{1}{\log n}$.

At step $j$, the algorithm queries node $q_{j+1}$ based on the information $K_j$ it has so far and its internal randomness $R$, then receives answer $a_{j+1}$. The the only way for both $A_j$ and $\neg A_{j+1}$ to hold is for the new answer $a_{j+1}$ to be smaller than $\lambda_{j+1}$.\footnote{Formally, if $A_j \wedge \neg A_{j+1}$, then using Fact~\ref{fact:lambda-decreasing} we have $\mu_{j+1} < \lambda_{j+1} < \lambda_j \leq \mu_j$. Therefore, $\mu_{j+1} = a_{j+1}$, and thus $a_{j+1} < \lambda_{j+1}$.}

Let $d\coloneqq |v^*-q_{j+1}|$. If we had $d \geq \frac{4}{3}\lambda_{j+1} \geq D$, then if $T$ occurs, by concentration bound (ii) we would have $a_{j+1} \geq \frac{3}{4}d \geq \lambda_{j+1}$. Let $I\coloneqq V \cap (q_{j+1} \pm \frac{4}{3}\lambda_{j+1})$ ($I$ also depends on $(K_j,R)$). Then the only way to have $a_{j+1} < \lambda_{j+1}$ is for $v^*$ to be in $I$, which implies
\begin{equation}
\label{eq:v-in-i}
\Pr[T \wedge A_j \wedge B_j \wedge \neg A_{j+1}] \leq \Pr[T \wedge v^* \in I \wedge B_j].
\end{equation}

Now, for any assignment $(k_j, r)$ of random variables $(K_j,R)$, we have
\begin{align*}
&\Pr[T \wedge v^* \in I \wedge B_j \mid K_j=k_j \wedge R=r]\\
%&\qquad= \Pr[T \wedge v^* \in I(k_j,r) \wedge B_j(k_j) \mid K_j=k_j \wedge R=r]\\
&\qquad= \sum_{v \in I} \Pr[T \wedge (v^*=v) \wedge B_j \mid K_j=k_j \wedge R=r]\tag{$I$ is fixed by $(K_j, R)$}\\
&\qquad= \sum_{v \in I} \Pr[T \wedge (v^*=v) \wedge B_j \mid K_j=k_j].\tag{$T,v^*$ are independent from $R$, and $B_j$ is fixed by $K_j$}
\end{align*}
If $B_j$ is false given $K_j=k_j$, then the above sum has probability $0$. If on the other hand $B_j$ is true given $K_j=k_j$, then by definition of $B_j$,
\begin{align*}
\sum_{v \in I} \Pr[T \wedge (v^*=v) \wedge B_j \mid K_j=k_j]
&= \sum_{v \in I(k_j,r)} \Pr[T \wedge (v^*=v) \mid K_j=k_j]\\
%&\qquad\leq \sum_{v \in I(k_j,r)} \dilution{1}\\
%&\qquad= \dilution{|I|}\\
%&\qquad= \frac{1}{\log n}.
&\leq \dilution{|I|} \leq \frac{1}{\log n}.
\end{align*}
Therefore, in either case, $\Pr[T \wedge v^* \in I \wedge B_j] \leq \frac{1}{\log n}$, which, combined with \eqref{eq:v-in-i}, completes the proof.
\end{proof}
% END \input{proofs/b-implies-a.tex}

\subsubsection{Induction on $j$}
Lemmas~\ref{lemma:a-implies-b} and~\ref{lemma:b-implies-a} can now be chained to obtain the following result.
\begin{lemma}
\label{lemma:a-jstop}
$\Pr[A_\jStop] \geq 1-p$.
\end{lemma}
\begin{proof}
First, as already noted at the end of Section~\ref{sss:mu}, $T$ implies $\mu_0 \geq \lambda_0$, which means that $T \Rightarrow A_0$ (by definition of $A_0$). Also, by Lemma~\ref{lemma:a-implies-b}, $A_j \Rightarrow B_j$ for $0 \leq j < \jStop$.
In addition, by Corollary~\ref{corollary:typical-is-likely}, we have $\Pr[T] \geq 1-p/2$, and by Lemma~\ref{lemma:b-implies-a}, for $0 \leq j < \jStop$, we have $\Pr[\neg T \vee \neg A_j \vee \neg B_j \vee A_{j+1}] \geq 1-\frac{1}{\log n}$. Therefore by a union bound, both $T$ and ``$\neg T \vee \neg A_j \vee \neg B_j \vee A_{j+1}$ for $0 \leq j < \jStop$'' simultaneously hold with probability at least
\begin{align*}
1-p/2 - \frac{\jStop}{\log n} \geq 1-p/2 - \frac{\floor*{\frac{p\log n}{2}}}{\log n}\geq 1-p.\tag{by Definition~\ref{def:jstop}}
\end{align*}
If they do hold, then the following logical statements are all true:
%\[T \qquad T \Rightarrow A_0 \qquad A_j \Rightarrow B_j\ (\forall\,j < \jStop) \qquad T \wedge A_j \wedge B_j \Rightarrow A_{j+1}\ (\forall\,j < \jStop).\]
``$T$'', ``$T \Rightarrow A_0$'', ``$A_j \Rightarrow B_j$'' ($\forall\,j < \jStop$), and ``$(T \wedge A_j \wedge B_j) \Rightarrow A_{j+1}$'' ($\forall\,j < \jStop$).
It is easy to see that, chained together, they imply $A_\jStop$.
%We prove by induction that if this is the case, then $A_\jStop$ holds. As already noted in Section~\ref{sss:mu}, if $T$ holds, then $\mu_0 \geq n/C =\lambda_0$, so $A_0$ holds. Now, assume that $A_j$ holds for some $0\leq j<\jStop$. By Lemma~\ref{lemma:a-implies-b}, this implies $B_j$. Then $T \wedge A_j \wedge B_j$ together imply $A_{j+1}$.
\end{proof}

\subsubsection{Asymptotics of $\jStop$}
The following lemma gives us an asymptotic lower bound on $\jStop$. We prove it in~\ref{proof:lemma:asymptotics-lb}.
\begin{lemma}
\label{lemma:asymptotics-lb}
For $n \geq \Theta_p(\max(\sigma^3, 1))$, we have
\[
\jStop+1 =
\begin{cases}
\Omega_{p}(1 + \log (1+\log_{1/\sigma} n))\text{ if $\sigma^2 \leq 1/2$}\\
\Omega_{p}(\log \log n)\text{ if $\sigma^2 \geq 1/2$.}
\end{cases}
\]
\end{lemma}

\subsubsection{Proof of Theorem~\ref{thm:sequential-lb}}
All that is left to do is to conclude.
\begin{proof}[Proof of Theorem~\ref{thm:sequential-lb}]
If $A_\jStop$ holds, then $\mu_\jStop > 0$, which means the algorithm has not found $v^*$ after asking $\jStop$ queries (recall our assumption from the beginning of Section~\ref{sec:sequential-lb} that, without loss of generality, the algorithm must query the source in order to make its guess). By Lemma~\ref{lemma:a-jstop}, this happens with probability at least $1-p$. Therefore, any algorithm that finds $v^*$ with probability at least $p$ must use at least $\jStop+1$ queries. The theorem then follows from Lemma~\ref{lemma:asymptotics-lb}.
\end{proof}
% END \input{sequential-lb.tex}
% END \input{sequential.tex}
% BEGIN \input{discussion.tex}
\section{Conclusion and Future Work}

We presented the first mathematical study of source identification with time queries in a non-deterministic diffusion process. We considered both the setting when the queries are selected adaptively and non-adaptively. We found that when the edge-delay distribution has constant variance, the number of required queries is $\Theta(\log \log n)$ in the adaptive setting, and $\Theta(n)$ in the non-adaptive setting. Our results are in sharp contrast with similar problems, such as measurement dependent noisy search on a line \cite{lalitha2017measurement}, or probabilistic binary search in graphs \cite{EmamjomehKS16}, where the query complexities were found to be $\Theta(\log n)$ in both cases. 

The main open question is of course what happens in other graphs. Extending our results to certain classes of trees might be feasible with the methods presented in this paper, however, an extension to graphs with cycles seems very challenging. Still, we hope that our results can inspire some, potentially more heuristic, ideas for treating graphs with cycles as well.

While we do not consider this scenario, given the sensitive nature of health information, it would be interesting to study source identification with time queries in the context of privacy preserving learning. In a scenario where an adversary is watching our queries, but not the responses, a recent line of work characterized the tradeoff between query complexity and privacy in adaptive binary search on a line \cite{tsitsiklis2018private,xu2018query}. The model has been extended to the case when the answers we receive are noisy in a follow-up work by \cite{xu2019optimal}. It would be interesting to combine the methods presented in the current paper with the methods of \cite{tsitsiklis2018private,xu2018query,xu2019optimal} for new results in privacy preserving source identification.
% END \input{discussion.tex}

\section{Acknowledgements}
The work presented in this paper was supported in part by the Swiss National Science Foundation under grant number 200021-182407. Most of Victor Lecomte's work was done while visiting EPFL on a Summer@EPFL fellowship.

\bibliographystyle{plainurl}% the mandatory bibstyle
\bibliography{cite}

\appendix
% BEGIN \input{other-distributions.tex}
\section{Extending to other edge-delay distributions}
\label{sec:other-distributions}

Throughout the paper, we assumed that the edge-delay distribution was Gaussian, however due to the Central Limit Theorem, it is natural to expect that our result generalizes to other distributions as well. However, there definitely are edge-delay distributions for which our result cannot generalize. Consider an edge-delay distribution $\W$ supported on two values: $1$ and $\pi$. Since $\pi$ is irrational, a single query node at one end of the path can determine the identity of the source with absolute certainty. Moreover, our results are not likely to generalize to heavy tailed $\W$ due to the lack of concentration in the answers.

We sketch how our proofs could be generalized to continuous sub-gaussian random variables. In the proofs of our main results, we exploit two types of properties of the edge-delay distribution; we are using the tight concentration of their sum in the non-adaptive upper bound and the adaptive upper and lower bounds, and we are using an anti-concentration result on their sum in the non-adaptive and adaptive lower bounds.

All of the concentration bounds are derived from Fact \ref{fact:normal}. This tail-bound result is easily extendable to sub-gaussian random variables (see Proposition 5.10 of \cite{vershynin_2012}). The only difference in the results would be that $\sigma$ would be replaced by the sub-gaussian norm 
$$\| \W\|_{\psi_2}~=~\mathrm{supp}_{p\ge1}p^{-1/2}(\Exp |\W|^p)^{1/p}.$$

In the adaptive lower bound proof, when we make the anti-concentration arguments, our proof uses the density function of the Gaussian distribution. Therefore, we need that the density function of $\sum \W_i$ is pointwise close to the density function of the corresponding Gaussian distribution. Such statements are called local limit theorems for sums of independent random variables. In a sense, we are asking for much more than a tail-bound, but we can also be much looser than an exponential decay. In Lemma \ref{lemma:diluted} we need that the probability mass function of the posterior is bounded above by $\log(\mu_j)/(\sigma\sqrt{\mu_j})$. We prove this by writing the probability mass function (as a function of potential source $v'$) explicitly using Bayes rule and the density of $\sum_{i=1}^{d_l'} \W_i$ and $\sum_{i=1}^{d_r'} \W_i$, where $d_l'=v'-l_j$ and $d_r'=r_j-v'$ are the distances between a node $v'$ and the closest query nodes to the left and the right. We need these densities to be pointwise $o(\log(\mu_j)/(\sigma\sqrt{\mu_j}))$ close to the densities in the Gaussian case. Since in Claim~\ref{claim:dl-dr-prime} we prove that $d_l' \in [1/2,2]a_{l_j}$ and $d_r' \in [1/2,2]a_{r_j}$, and by the definition $\mu_j=\min(a_{l_j},a_{r_j})$, it is enough to show that the density of $\sum_{i=1}^{d_l'} \W_i$ is pointwise $o(\log(d_l')/(\sigma\sqrt{d_l'}))$ close to the density in the Gaussian case (and we need the symmetric statement for $d_r'$). Such results are readily available for continuous distributions $\W$ with finite third moment (see Theorem 7.15 in \cite{petrov2012sums}). We also point out, that similar results exist for discrete distributions $\W$ satisfying a certain lattice condition that can be used to rule out distributions like the one supported on $1$ and $\pi$ that we used as a counterexample in the beginning of the section (see Theorem 7.6 in \cite{petrov2012sums}).

For the the anti-concentration result in the non-adaptive lower bound, we proved that the hypothesis testing problem cannot be solved between $k$ neighboring nodes at distance $d$ or more away from the query nodes. For this we upper bounded the union of the area under the probability density functions of the answers under each of the $k$ hypotheses by another another function, which we could easily integrate. For general edge-delay distributions, again we aim to approximate the probability density functions of the answers by the the probability density function of Gaussian random variables. Since this time, instead of small $l_\infty$ distance, we need small $l_1$ distance between the densities, a Berry-Esseen type theorem \cite{berry1941accuracy,esseen1945fourier} suffices instead of a local limit theorem.

We note that only the concentration arguments required the sub-gaussianity of the edge-delay distribution, the anti-concentration results held for a much more general class of distributions (finite third moment and continuity or lattice condition). We believe that with more advanced proof techniques the sub-gaussianity condition can also be relaxed.
% END \input{other-distributions.tex}
% BEGIN \input{SD1.tex}
\section{The difference between S1 and S2}
\label{appendix:SD1}

The only difference between two models S1 and S2 defined in Section~\ref{sec:intro} is that the starting time of the epidemics is unknown in S1 and known in S2. We already mentioned that S2 is theoretically more appealing, and that there is little difference between the number of queries required in the two models. The main consequence of the difference between the source identification algorithms in the two models is that in S1, the answers that they can use are the relative differences between time measurements at different pairs of query nodes, whereas in S2 the answers they can use are the absolute differences between the (known) starting time of the epidemics and the time measurement at each query node. Since S1 is more restrictive than S2, our lower bounds on the number of required queries in S2 clearly also hold in S1. We comment on how the upper bounds can be extended in Remarks \ref{rem:SD1_static} and \ref{rem:SD1_sequential}.

Additionally, we argue that while S2 has a simpler mathematical definition than S1, on the path network, proving lower bounds for S2 raises important challenges that would not have appeared in S1. Indeed, in the path network, the pair of query nodes that surround the source provide two independent answers about it in S2 (one from each direction between each query node and the source), but only one in S1 (because only the time difference between the measurements is meaningful). As a result, the analysis of the required number of queries is more challenging in S2 than in S1 because of the richer set of independent answers. Incorporating several independent measurements will be the main difficulty for the analysis of the number of queries needed to identify the source in more complex network models, such as bounded-degree trees. By focusing on S2 in the path network, our paper therefore paves the way towards the analysis of 
%problem (iii) on these
more complex network models.

\begin{remark}
\label{rem:SD1_static}
If the time of the first infection is not known (model S1), we can model the answers by adding an unknown constant $T_{start}$ to all of them. Then, Claim \ref{claim:end-of-static-ub} (a) and (b) hold without modification and we can prove a version of (c) where the differences of the distances equal the differences of the answers rounded to the nearest integer (we just need a slightly tighter concentration result). Let us also define $\sRight \coloneqq \sSmall+d$, and $\oRight$ as the corresponding answer. Then, by Claim \ref{claim:end-of-static-ub}, if $\round{\oSmall-\oLeft}=d$ then $v^*$ is between $\sSmall$ and $\sRight$ and we can find $v^*$ by computing
$$\frac{\round{\oSmall-\oRight}+\sSmall+\sRight}{2}=\frac{(v^*-\sSmall)-(\sRight-v^*)+\sSmall+\sRight}{2}=v^*$$
Otherwise, if $\round{\oSmall-\oLeft}<d$ then $v^*$ is between $\sLeft$ and $\sSmall$, and $v^*$ can be found analogously.
\end{remark}

\begin{remark}
\label{rem:SD1_sequential}
If the time of the first infection is not known (model S1), and we model the answers by adding an unknown constant $T_{start}$ to all of them, then a version of Lemma \ref{lemma:concentration-ub} shifted by $T_{start}$ still holds. In this case, a slightly modified version of the algorithm finds the source. At each step the algorithm will query two nodes: one at $l_i$ and $r_i$, with $l_0=1$ and $r_0=n$. Then, we have a similar equation as \eqref{eq:seq_upper_interval} for the difference of the answers
\begin{align*}
\dw(v^*,l_i)-\dw(v^*,r_i) &\in (v^*-l_i)-(r_i-v^*) \pm 2C(\delta)\cdot \sigma\sqrt{d_i}\ln(1+d_i),
\end{align*}
where $d_i \coloneqq r_i-l_i$, which means that we can keep track of a shrinking interval
\[
\begin{cases}
l_{i+1} = \max\left(l_i,\left\lceil \frac{\dw(v^*,l_i)-\dw(v^*,r_i) +l_i +r_i}{2} - C(\delta)\cdot\sigma\sqrt{d_i}\ln(1+d_i) \right\rceil \right)\\
r_{i+1} = \min\left(r_i, \left\lfloor \frac{\dw(v^*,l_i)-\dw(v^*,r_i) +l_i +r_i}{2} + C(\delta)\cdot\sigma\sqrt{d_i}\ln(1+d_i) \right\rfloor \right).
\end{cases}
\]
The rest of the proof can be written similarly to the case when the time of the first infection is known, and the only change in the final result is that we used twice as many queries to identify the source (which does not affect the asymptotic results).
\end{remark}
% END \input{SD1.tex}
% BEGIN \input{sim_fig.tex}

\section{Simulation details for Figure \ref{fig:sim_fig}}
\label{appendix:sim_fig}

The simulation results were generated in the S1 source identification model with the Python toolbox \cite{spinelliGithub}, which has been published in \cite{spinelli2017back}. The underlying diffusion process was a Susceptible-Infected process (also called First Passage percolation) with uniform edge-delay distribution supported on the interval $[1-\sqrt{3}\sigma, 1+\sqrt{3}\sigma]$ (with mean $1$ and standard deviation $\sigma$). Thereafter, a uniformly random query node was picked, and all further queries were selected by the Max-Gain algorithm as implemented in \cite{spinelliGithub}. The algorithm was stopped when the candidate set reduced to a single node, which always had to be the source, since the Max-Gain algorithm always finds the correct source if enough queries are provided and the edge-delay distribution is bounded in some interval $[1-\epsilon, 1+\epsilon]$ for $\epsilon \in (0,1)$ (see Theorem 2 of \cite{spinelli2017back}). The number of queries plotted in Figure~\ref{fig:sim_fig} is the number of queries used by the Max-Gain until it was stopped, averaged over 192 simulations.

% END \input{sim_fig.tex}
% BEGIN \input{annoying-proofs.tex}
%\section{Technical proofs}
%\input{proofs/normal}
% BEGIN \input{proofs/end-static-ub.tex}
\section{Proof of Claim~\ref{claim:end-of-static-ub}}
\label{proof:claim:end-of-static-ub}

\begin{repclaim}{claim:end-of-static-ub}
For some $d = \Omega\left(\frac{1}{\sigma^2\log(1/\delta)}\right)$,
all of the following hold simultaneously with probability $\geq 1-\delta$:
\begin{alphaenumerate}
\item among the query nodes located at or to the left of $v^*$, the closest one is the one with the smallest answer;
\item among the query nodes located at or to the right of $v^*$, the closest one is the one with the smallest answer;
\item the two closest query nodes to $v^*$ on its left side and the closest query node on its right side all give a correct answer once rounded to the nearest integer.
\end{alphaenumerate}
\end{repclaim}

\begin{proof}%[Proof of Claim~\ref{claim:end-of-static-ub}]
In this proof, we will assume that
\begin{equation}
\label{eq:var-small-enough}
\sigma^2 \leq \frac{1}{16\ln(6/\delta)} \leq \frac{1}{2\ln(12/\delta)}.
\end{equation}
If this is not the case, then $\sigma^2 = \Omega(1/\log(1/\delta))$, so we can simply query every node, which gives $d=1=\Omega\left(\frac{1}{\sigma^2\log(1/\delta)}\right)$.

We need to choose $d$ such that wherever $v^*$ is identified, (a), (b), (c) simultaneously hold with probability $\geq 1-\delta$ over the choice of the weights $w(\cdot)$. Let us first study point (b) (point (a) is analogous). Let $q$ be the closest query node at or to the right of $v^*$. Then (b) is true iff
\begin{itemize}
\item the sum of the weights of the edges between $q$ and $q+d$ is positive;
\item the sum of the weights of the edges between $q$ and $q+2d$ is positive;
\item \ldots
\item the sum of the weights of the edges between $q$ and $1+\floor*{\frac{n-1}{d}}d$ is positive.
\end{itemize}
More formally, (b) is true iff for all integers $i>0$ such that $q+id\leq n$, the sum of the weights of the edges between $q$ and $q+id$ is positive.

A sufficient condition for this to hold is: for all positive integers $x$, the sum of the weights of the edges between nodes $q$ and $q+x$ is positive. By Fact~\ref{fact:sum-normal}, each of these sums is distributed as a Gaussian $\Normal(x, x\sigma^2)$, so it is positive except with probability
\begin{align*}
\Pr_{X \sim \Normal(x, x\sigma^2)}[X < 0]
&\leq \Pr_{X \sim \Normal(x, x\sigma^2)}[X \notin x \pm x]\\
&\leq e^{-\frac{x^2}{2x\sigma^2}}\tag{Fact~\ref{fact:normal}}\\
&= e^{-\frac{x}{2\sigma^2}}.
\end{align*}
Therefore, by a union bound, (b) holds except with probability at most
\begin{align*}
\sum_{x=1}^\infty \left(e^{-\frac{1}{2\sigma^2}}\right)^x
&= \frac{e^{-\frac{1}{2\sigma^2}}}{1-e^{-\frac{1}{2\sigma^2}}}\\
&< 3e^{-\frac{1}{2\sigma^2}}\tag{because $\sigma^2 \leq 1 \Rightarrow e^{-\frac{1}{\sigma^2}} < 2/3$}
\end{align*}
which, assuming $\sigma^2 \leq \frac{1}{2\ln(12/\delta)}$ (equation \eqref{eq:var-small-enough}), is at most $\delta/4$. 

Finally, we study the probability that (c) holds. Let $d_1,d_2,d_3$ be the distances of those three query nodes to $v^*$. They are all at most $2d$ away from $v^*$. For $i=1,2,3$, the corresponding answer is distributed as $X \sim \Normal(d_i, d_i\sigma^2)$, and is correct after rounding iff $X \in (d_i-1/2, d_i+1/2)$. Therefore, (c) holds except with probability
\begin{align*}
\sum_{i=1}^3\Pr_{X \sim \Normal(d_i, d_i\sigma^2)}[X \notin d_i \pm 1/2]
&\leq \sum_{i=1}^3 e^{-\frac{1}{8d_i\sigma^2}}\tag{Fact~\ref{fact:normal}}\\
&\leq 3e^{-\frac{1}{16d\sigma^2}}
\end{align*}
which, assuming $d \leq \frac{1}{16\sigma^2 \ln(6/\delta)}$, is at most $\delta/2$. Therefore, we set $d \coloneqq \floor*{\frac{1}{16\sigma^2 \ln(6/\delta)}}$. By \eqref{eq:var-small-enough}, $\frac{1}{16\sigma^2 \ln(6/\delta)} \geq 1$, so $d \geq (1/2) \cdot \frac{1}{16\sigma^2 \ln(6/\delta)} = \Omega\left(\frac{1}{\sigma^2\log(1/\delta)}\right)$, as required. 

Finally, by one more union bound, for our chosen value of $d$, all of (a), (b), (c) hold except with probability at most $\delta/4 + \delta/4 + \delta/2 = \delta$.
\end{proof}
% END \input{proofs/end-static-ub.tex}
%\input{proofs/tvd}
% BEGIN \input{proofs/concentration-ub.tex}
\section{Proof of Lemma~\ref{lemma:concentration-ub}}
\label{proof:lemma:concentration-ub}

\begin{replemma}{lemma:concentration-ub}
For any probability $0 <\delta < 1/2$, there is some constant $C(\delta)=O\left(\sqrt{\log(1/\delta)}\right)$ such that for any $n,\sigma$ and any source $v^* \in V$, we have
\begin{equation}\label{eq:concentration-ub-restate}
\Pr_w[\forall q \in V, \,\dw(v^*,q) \in |v^*-q| \pm C(\delta)\cdot\sigma\sqrt{|v^*-q|}\ln (1+|v^*-q|)] \geq 1-\delta.
\end{equation}
\end{replemma}

\begin{proof}%[Proof of Lemma~\ref{lemma:concentration-ub}]

We will use the quantity $|v^*-q|$ many times in this proof, so to simplify notation, let $d_{q} \coloneqq |v^*-q|$.
We will fix $C(\delta)$ later, but for the moment assume $C(\delta)\geq 2$ (this is clearly the case for the value we set it to in \eqref{eq:value-of-c}).

First of all, for $q=v^*$, \eqref{eq:concentration-ub-restate} holds trivially.
For any other $q$, by Fact~\ref{fact:sum-normal}, $\dw(v^*, q)$ is distributed according to Gaussian $\Normal(d_{q},d_{q}\sigma^2)$ (though not independently).
Then, by Fact~\ref{fact:normal}, for any $q \neq v^*$,
\begin{align*}
\Pr[\dw(v^*,q) \notin d_{q} \pm C(\delta)\sigma\sqrt{d_{q}} \ln (1+d_{q})]
&\leq e^{-\frac{(C(\delta)\sigma\sqrt{d_{q}}\ln (1+d_{q}))^2}{2d_{q}\sigma^2}}\\
&= e^{-\frac{C(\delta)^2(\ln(1+d_{q}))^2}{2}}\\
&= \left(\frac{1}{1+d_{q}}\right)^{\frac{C(\delta)^2\ln(1+d_{q})}{2}}\\
&= \left(\frac{1}{1+d_{q}}\right)^{\frac{(C(\delta)^2-1)\ln(1+d_{q}) + \ln(1+d_{q})}{2}}\\
&= \left(\frac{1}{1+d_{q}}\right)^{\frac{(C(\delta)^2-1)\ln(1+d_{q})}{2}} \left(\frac{1}{1+d_{q}}\right)^{\frac{\ln(1+d_{q})}{2}}\\
&\leq \left(\frac{1}{1+d_{q}}\right)^{\frac{C(\delta)^2}{4}} \left(\frac{1}{1+d_{q}}\right)^{\frac{\ln(1+d_{q})}{2}}\tag{$C(\delta) \geq 2$ so $(C(\delta)^2-1)\ln(2) \geq C(\delta)^2/2$}\\
&\leq \left(\frac{1}{2}\right)^{\frac{C(\delta)^2}{4}} \left(\frac{1}{1+d_{q}}\right)^{\frac{\ln(1+d_{q})}{2}}.
\end{align*}
By a union bound over all $q \neq v^*$, this implies that
\begin{align}
\Pr[\exists\,q \in V, \,\dw(v^*,q) \notin d_{q} \pm C(\delta)\sigma\sqrt{d_{q}} \ln (1+d_{q})]
&\leq 2\sum_{d=1}^\infty \left(\frac{1}{2}\right)^{\frac{C^2}{4}} \left(\frac{1}{1+d}\right)^{\frac{\ln(1+d)}{2}}\nonumber\\
&= 2\left(\frac{1}{2}\right)^{\frac{C^2}{4}} \sum_{d=1}^\infty\left(\frac{1}{1+d}\right)^{\frac{\ln(1+d)}{2}}\nonumber\\
&= 2\left(\frac{1}{2}\right)^{\frac{C^2}{4}} \left(\sum_{d=1}^{\ceil*{e^4+1}-1}\left(\frac{1}{1+d}\right)^{\frac{\ln(1+d)}{2}} + \sum_{d=\ceil*{e^4+1}}^\infty\left(\frac{1}{1+d}\right)^{\frac{\ln(1+d)}{2}}\right)\nonumber\\
&\leq 2\left(\frac{1}{2}\right)^{\frac{C^2}{4}} \left(\ceil*{e^4+1}-1 + \sum_{d=\ceil*{e^4+1}}^\infty\left(\frac{1}{1+d}\right)^{2}\right)\tag{$d \geq e^4-1$ implies $\ln(1+d)/2 \geq 2$}\\
&\leq 2\left(\frac{1}{2}\right)^{\frac{C^2}{4}} \left(\ceil*{e^4} + \sum_{d=1}^\infty\frac{1}{d^2}\right)\nonumber\\
&= 2\left(\frac{1}{2}\right)^{\frac{C^2}{4}} \left(\ceil*{e^4} + \frac{\pi^2}{6}\right)\nonumber\\
&\leq 114\left(\frac{1}{2}\right)^{\frac{C^2}{4}}
\label{eq:product-series}
\end{align}

Setting
\begin{equation}
\label{eq:value-of-c}
    C(\delta) \coloneqq \sqrt{4 \log_2(114/\delta)} = O\left(\sqrt{\log(1/\delta)}\right),
\end{equation}
\eqref{eq:product-series} becomes $\leq \delta$, and we are done.

\end{proof}
% END \input{proofs/concentration-ub.tex}
% BEGIN \input{proofs/asymptotics-ub.tex}
\section{Proof of Claim~\ref{claim:asymptotics-ub}}
\label{proof:claim:asymptotics-ub}

\begin{repclaim}{claim:asymptotics-ub}
Assume $d_0\leq n$, and $d_{i+1} \leq C \cdot\sigma\sqrt{d_i}\ln(1+d_i)$ for some value $C>0$. Then
\begin{itemize}
\item if $\sigma^2 \leq 1/2$, there exists $k = O(\log \log_{1/\sigma} n)$ such that $d_k = \poly(C)$;
\item if $\sigma^2 \geq 1/2$, there exists $k = O(\log \log n)$ such that $d_k = \sigma^2\cdot \poly(C, \log(1+\sigma^2))$.
\end{itemize}
\end{repclaim}

\begin{proof}%[Proof of Claim~\ref{claim:asymptotics-ub}]
We track the value of $d_i/\sigma^2$ as $i$ increases. First, as long as
\begin{equation}
\label{eq:limit-di}
d_i \geq \sigma^2(C\ln(1+d_i))^6
\end{equation}
we have
\[
\frac{d_{i+1}}{\sigma^2} \leq \frac{C\sigma\sqrt{d_i}\ln(1+d_i)}{\sigma^2} = C \sqrt{\frac{d_i}{\sigma^2}}\ln(1+d_i) \stackrel{\eqref{eq:limit-di}}{\leq} \left(\frac{d_i}{\sigma^2}\right)^{2/3}
\]
(the last inequality can be deduced by dividing both sides by $\sqrt{\frac{d_i}{\sigma^2}}$ then raising both sides to the sixth power).
Thus, by induction, as long as \eqref{eq:limit-di} holds, we have
\[\frac{d_i}{\sigma^2} \leq \left(\frac{d_0}{\sigma^2}\right)^{(2/3)^i} \Rightarrow d_i \leq \sigma^2\left(\frac{n}{\sigma^2}\right)^{(2/3)^i}.\]

\newcommand{\Dmin}{d_\text{min}}
Let $\Dmin$ be the smallest value greater than $\max(2\sigma^2, 1)$ that we can assign to $d_i$ such that \eqref{eq:limit-di} holds. Let $k$ be the smallest integer for which $d_k \leq \Dmin$. Then we have
\[\sigma^2 \left(\frac{n}{\sigma^2}\right)^{(2/3)^{k-1}} \geq \Dmin
\Leftrightarrow
k
\leq 1+ \log_{3/2} \left(\frac{\log\left(\frac{n}{\sigma^2}\right)}{\log\left(\frac{\Dmin}{\sigma^2}\right)}\right)
\leq 1+\log_{3/2}\left(\frac{\log\left(\frac{n}{\sigma^2}\right)}{\log(\max(1/\sigma^2, 2))}\right).\]
\begin{itemize}
\item If $\sigma^2 \leq 1/2$, then it is easy to verify that $\Dmin = O((C\log C)^6) = \poly(C)$. Therefore, for $k = O(1+ \log(1+\log_{1/\sigma}n))$, we have $d_k \leq \Dmin = \poly(C)$.
\item If $\sigma^2 \geq 1/2$, then it is easy to verify that $\Dmin = O(\sigma^2C^6\log(1+\sigma^2C^6)^6) = \sigma^2\cdot \poly(C, \log(1+\sigma^2))$.
%\footnote{Note that we did not try to optimize the asymptotic bound for $\Dmin$ in terms of $\sigma$: it could be reduced to just $O(\sigma^2)$.
%This doesn't really matter, because typical values of $\sigma$ are not very big.
%This doesn't really matter, because (a) typical values of $\sigma$ are not very big, and (b) as we mention later, in most applications the edge distribution is positive, which allows running a quick $O(\log \sigma)$ binary search to finish up.}
Therefore, for $k = O(\log\log n)$, we have $d_k \leq \Dmin = \sigma^2 \cdot \poly(C, \log(1+\sigma^2))$.\qedhere
\end{itemize}
\end{proof}
% END \input{proofs/asymptotics-ub.tex}
% BEGIN \input{proofs/typical-is-likely.tex}
\section{Proof of Lemma~\ref{lemma:typical-is-likely}}

\label{proof:lemma:typical-is-likely}
\begin{replemma}{lemma:typical-is-likely}
For any probability $\delta>0$ and any $n \geq \Theta_\delta(\max(\sigma^2\ln\sigma, 1))$,
$\Pr_{v^*,w}[\Typical_{\delta}(v^*,w)] \geq 1-\delta$.
\end{replemma}

Before proving Lemma~\ref{lemma:typical-is-likely}, we first prove two claims.
\begin{claim}
\label{claim:conc1}
For any probability $\delta_1>0$, there exists $D_1(\delta_1)>0$ such that for any $n,\sigma$ and any source $v^* \in V$,
\[
\Pr_w\left[\text{for all $q$ such that }d_{q} \coloneqq |v^*-q| \geq D_1(\delta_1),\,\dw(v^*,q) \in d_{q} \pm \sigma \sqrt{d_{q}} \ln d_{q}\right] \geq 1-\delta_1.
\]
\end{claim}

\begin{claimproof}
At first, let us consider only the case $q \geq v^*$. That is, consider node $q = v^* + d$ for some distance $d \geq e^2$. Then $\dw(v^*,q) \sim \Normal(d, d\sigma^2)$, so
\begin{align*}
\Pr_w\big[\dw(v^*,q) \in d \pm \sigma\sqrt{d} \ln d\big]
&\geq 1-e^{\frac{(\sigma\sqrt{d} \ln d)^2}{2d\sigma^2}}\tag{from Fact~\ref{fact:normal}}\\
&= 1-e^{-\frac{(\ln d)^2}{2}}\\
&= 1- \frac{1}{d^{(\ln d) / 2}}.
\end{align*}
Now, for any integer $D_1\geq e$, by a union bound, this will hold for all $q \geq v^* + D_1$ with probability at least
\[1 - \sum_{d=D_1}^\infty \frac{1}{d^{(\ln d) / 2}}.\]
Note that this sum converges, because $(\ln d)/2 > 1$ for large enough $d$. Thus the sequence of sums $(\sum_{d=k}^\infty 1/d^{(\ln d)/2})_{k \geq 3}$ converges to $0$ and we can define $D_1(\delta_1) \coloneqq \min \{k \geq 3 \mid \sum_{d=k}^\infty 1/d^{(\ln d)/2} \leq \delta_1/2\}$.

Therefore, by going through the same reasoning for $q \leq v^*$ and taking a union bound, we get that
\[\dw(v^*,q) \in d_{q} \pm \sigma \sqrt{d_{q}} \ln d_{q}\]
will hold for all $q$ at distance $d_{q} \coloneqq |v^*-q| \geq D_1(\delta_1)$, except with probability at most $\delta_1/2+\delta_1/2=\delta_1$.
\end{claimproof}

\newcommand{\Dtwo}{D_2(\delta_2,\epsilon, \sigma)}
\begin{claim}\label{claim:conc2}
For any probability $\delta_2 > 0$ and any $\epsilon \in (0,1)$, there exists $\Dtwo>0$ such that for any $n,\sigma$ and any source $v^* \in V$,
\[
\Pr_w\left[\text{for all $q$ such that }d_{q} \coloneqq{} |v^*-q| \geq \Dtwo,\,\dw(v^*,q) \in (1\pm \epsilon)d_{q}\right] \geq 1-\delta_2,
\]
and $\Dtwo = O_{\delta_2,\epsilon}(\max(\sigma^2\log\sigma, 1))$.
\end{claim}

\begin{claimproof}
At first, let us consider only the case $q \geq v^*$. That is, consider node $q = v^* + d$ for some distance $d$. Then $\dw(v^*,q) \sim \Normal(d, d\sigma^2)$, so
\begin{align*}
\Pr_w\big[\dw(v^*,q) \in (1\pm \epsilon)d\big]
&\geq 1-e^{\frac{(\epsilon d)^2}{2d\sigma^2}}\tag{from Fact~\ref{fact:normal}}\\
&= 1-(e^{-\frac{\epsilon^2}{2\sigma^2}})^d.
\end{align*}
Now, for any integer $D_2$, by a union bound, this will hold for all $q \geq v^* + D_2$ except with probability at most
\begin{align}
\label{eq:prob-D2}
\sum_{d=D_2}^\infty \left(e^{-\frac{\epsilon^2}{2\sigma^2}}\right)^d
&=e^{-\frac{\epsilon^2}{2\sigma^2}D_2}\sum_{d=0}^\infty \left(e^{-\frac{\epsilon^2}{2\sigma^2}}\right)^d\nonumber\\
&= \frac{e^{-\frac{\epsilon^2}{2\sigma^2}D_2}}{1-e^{-\frac{\epsilon^2}{2\sigma^2}}},
\end{align}
where the last step uses the fact that this is a geometric series.
\begin{itemize}
\item If $\frac{\epsilon^2}{2\sigma^2} \geq 1$, then $\eqref{eq:prob-D2} \leq \frac{e^{-D_2}}{1-1/e}$, so if we set $\Dtwo \coloneqq \ln\left(\frac{2}{\delta_2(1-1/e)}\right)$, then $\eqref{eq:prob-D2} \leq \delta_2/2$.
\item If $\frac{\epsilon^2}{2\sigma^2} \leq 1$, then we can use $e^{-x} \leq 1-x/2$ on $[0,1]$ to obtain that
\[\eqref{eq:prob-D2} \leq \frac{2e^{-\frac{\epsilon^2}{2\sigma^2}D_2}}{\frac{\epsilon^2}{2\sigma^2},}\]
so if we set $\Dtwo \coloneqq \frac{2\sigma^2}{\epsilon^2}\ln\left(\frac{4\sigma^2}{\epsilon^2\delta_2}\right)$, then $\eqref{eq:prob-D2} \leq \delta_2/2$.
\end{itemize}
It is easy to check that both these values are $O_{\delta_2,\epsilon}(\max(\sigma^2\ln\sigma, 1))$.

Finally, by going through the same reasoning for $q \leq v^*$ and taking a union bound, we get that
\[\dw(v^*,q) \in (1\pm \epsilon)d_{q}\]
will hold for all $q$ at distance $d_{q} \coloneqq |v^*-q| \geq \Dtwo$, except with probability at most $\delta_2/2+\delta_2/2=\delta_2$.
\end{claimproof}

\begin{definition}
\label{def:cd}
Let $C(\delta) \coloneqq 8/\delta$ and $D(\sigma,\delta) \coloneqq \max(D_1(\delta/3),D_2(\delta/3,1/4,\sigma), \sigma^2, e^2)$.
\end{definition}
Let us verify that this definition of $D(\sigma, \delta)$ satisfies the bounds claimed in equation~\eqref{eq:guarantees-on-cd}.
%\begin{proof}[]
The lower bound of $\max(\sigma^2, e^2)$ is trivial. The upper bound of $O_\delta(\max(\sigma^2\log\sigma, 1))$ comes from the fact that $\Dtwo = O_{\delta_2,\epsilon}(\max(\sigma^2\log\sigma, 1))$.
%\end{proof}

We can now prove Lemma~\ref{lemma:typical-is-likely}.
\begin{proof}[Proof of Lemma~\ref{lemma:typical-is-likely}]
Apply Claim~\ref{claim:conc1} with $\delta_1\coloneqq\delta/3$, and Claim~\ref{claim:conc2} with $\delta_2\coloneqq\delta/3$ and $\epsilon\coloneqq 1/4$. Assume $n \geq 12/\delta$, and let $C'\coloneqq 6/\delta$. Since $v^*$ is uniformly distributed over $V=[n]$, we have
\[\Pr[\min(|v^*-1|,|v^*-n|) < n/C'] \leq \frac{2+n/C'}{n} \leq \frac{2}{n} + \frac{1}{C'} \leq \delta/3.\]
By a union bound, the concentration bounds of both Claim~\ref{claim:conc1} and Claim~\ref{claim:conc2} as well as inequality $\min(|v^*-1|,|v^*-n|) \geq n/C'$ will all hold with probability at least $1-\delta/3-\delta/3-\delta/3 = 1-\delta$.

Furthermore, by the concentration bound of Claim~\ref{claim:conc2}, if $\min(|v^*-1|,|v^*-n|) \geq n/C' \geq D_2(\delta/3,1/4)$, then
\begin{equation}
\label{eq:typical-a}
\min(\dw(v^*,1),\dw(v^*,n)) \geq (n/C')(1-1/4) = \frac{3n}{4C'} = \frac{n}{C(\delta)}.
\end{equation}
%Let $C(\delta) \coloneqq \frac{4C'}{3} = 8/\delta$, and let $D(\sigma,\delta) \coloneqq \max(D_1(\delta/3),D_2(\delta/3,1/4,\sigma))$.
Then for $n \geq \max(12/\delta, C'D_2(\delta/3,1/4,\sigma)) = O_\delta(\max(\sigma^2\ln\sigma, 1))$, with probability at least $1-\delta$, we have
\begin{itemize}
\item $\min(\dw(v^*,1),\dw(v^*,n)) \geq \frac{n}{C(\delta)}$ (from \eqref{eq:typical-a});
\item for all $q$ with $d_{q} \coloneqq |v^*-q| \geq D(\sigma,\delta)$,
\begin{itemize}
\item $\dw(v^*,q) \in d_{q} \pm \sigma \sqrt{d_{q}}\ln d_{q}$ (from Claim~\ref{claim:conc1})
\item $\dw(v^*,q) \in d_{q}(1 \pm 1/4)$ (from Claim~\ref{claim:conc1}).
\end{itemize}
\end{itemize}
\end{proof}
% END \input{proofs/typical-is-likely.tex}
% BEGIN \input{proofs/lambda-decreasing.tex}
\section{Proof of Fact~\ref{fact:lambda-decreasing}}
\label{proof:fact:lambda-decreasing}
\begin{repfact}{fact:lambda-decreasing}
For $0 \leq j \leq \jStop$, $\lambda_{j+1} < \lambda_j$.
\end{repfact}
\begin{proof}[Proof of Fact~\ref{fact:lambda-decreasing}]
Since $j \leq \jStop$, by definition of $\jStop$, $\lambda_j \geq D$. Also, by equation~\eqref{eq:guarantees-on-cd}, $D \geq \max(\sigma^2, e^2)$. Therefore,
\begin{align*}
\lambda_{j+1}
&=\reduce(\lambda_j)\tag{Definition~\ref{def:lambda}}\\
&=\frac{\sigma\sqrt{\lambda_j}}{400\ln\lambda_j\log n}\tag{Definition~\ref{def:reduce}}\\
&<\sigma\sqrt{\lambda_j}\tag{$n \geq 3$, $\lambda_j \geq e^2$}\\
&\leq \lambda_j.\tag{$\lambda_j \geq \sigma^2$}
\end{align*}
\end{proof}
% END \input{proofs/lambda-decreasing.tex}
% BEGIN \input{proofs/a-implies-b.tex}
\section{Proof of Lemma~\ref{lemma:a-implies-b}}
\label{proof:lemma:a-implies-b}

\begin{replemma}{lemma:a-implies-b}
If $j < \jStop$, then $A_j \Rightarrow B_j$.
\end{replemma}

The first step in proving Lemma~\ref{lemma:a-implies-b} is to prove that only a small part of the information contained in $K_j$ will actually influence the posterior of $v^*$ given $K_j$: only the closest query nodes to the source $q_{l_j}, q_{r_j}$ and the corresponding answers $a_{l_j}, a_{r_j}$ will have an influence (they are introduced in Definition~\ref{def:mu}).

\begin{definition}[$E_{l,r,x,y}$]
For any $l,r,x,y$, let $E_{l,r,x,y}$ be the event that $v^* \in [l,r]$, $\dw(v^*,l) = x$, and $\dw(v^*,r) = y$.
\end{definition}
Note that event $E_{l,r,x,y}$ depends purely on $v^*$ and $w$, not on the actions of the algorithm.

\begin{lemma}\label{lemma:only-lr}
Recall that $R$ is the internal randomness of the algorithm (see Definition~\ref{def:r}). For any node $v \in V$,
\begin{equation*}
\Pr_{v^*,w,R}[v^*=v \mid K_j] = \Pr_{v^*,w}[v^* = v \mid \GroupComplex].
\end{equation*}
%In other words, if we know between which two sensors $l$ and $r$ the source $v^*$ lies, then the distribution of $v^*$ conditioned on all observations is the same as conditioned on only the closest observations $a_l$ and $a_r$.
\end{lemma}

Before proving this lemma, we need to show a simple property of independence and conditional probability.

%%%%%%%%%%%%%%%%%%
%%%%%% FACT %%%%%% (weird conditioning fact)
%%%%%%%%%%%%%%%%%%

\begin{fact}\label{fact:conditioning}
Let $X,Y$ be independent random variables. Let $E(X),F(X)$ be Boolean functions depending only on $X$, and let $G(F(X),Y)$ be a Boolean function depending only on $F(X)$ and $Y$.
For simplicity, let us use $E(X)$ to denote the event $E(X)=1$ (and similarly for $F,G$).
Then we have
\[\Pr[E(X) \mid F(X) \wedge G(F(X),Y)] = \Pr[E(X) \mid F(X)].\]
\end{fact}

\begin{proof}
Intuitively, the reason this is true is that as $G$ depends only on $F$ (which is already provided in the conditioning) and $Y$ (which is independent from $X$), adding $G$ to the conditioning does not bring more information towards figuring out whether $E$ will happen or not. Formally, let $G'(Y) = G(1,Y)$ (1 represents ``true''). Then
\begin{align*}
\Pr[E(X) \mid F(X) \wedge G(F(X),Y)]
&= \Pr[E(X) \mid F(X) \wedge G'(Y)]\\
&= \frac{\Pr[E(X) \wedge F(X) \wedge G'(Y)]}{\Pr[F(X) \wedge G'(Y)]}\\
&= \frac{\Pr[E(X) \wedge F(X)] \cdot \Pr[G'(Y)]}{\Pr[F(X)]\cdot \Pr[G'(Y)]}\tag{independence of $X$ and $Y$}\\
&= \frac{\Pr[E(X) \wedge F(X)]}{\Pr[F(X)]}\\
&= \Pr[E(X) \mid F(X)]
\end{align*}
\end{proof}

% PROOF OF THE LEMMA

\newcommand{\One}{\mathbbm{1}}
We will notation $\One[\cdots]$ to denote the indicator Boolean function corresponding to some expression.

\begin{proof}[Proof of Lemma~\ref{lemma:only-lr}]
Fix any possible assignment $k_j$ for $K_j$.
We observe that event ``$K_j = k_j$'' is entirely determined by
\begin{alphaenumerate}
\item $R$ (the internal randomness of the algorithm);
\item the values of $\dw(v^*,q_i)$ for each $i \leq j$ (the answers to the first $j$ queries);
\item whether $v^* \in [q_{l_j},q_{r_j}]$.
\end{alphaenumerate}
Conversely, (b) and (c) are entirely determined by $K_j$. In addition, (b) and (c) depend only on $v^*$ and $w$, which are independent from $R$.

By the above, we can use Fact~\ref{fact:conditioning}, plugging in $E((v^*,w)) \coloneqq \One[v^*=v]$, $F((v^*,w)) \coloneqq \One[v^* \in [q_{l_j},q_{r_j}] \wedge \dw(v^*,q_i)=a_i \,\forall i \geq j]$ and $G(F((v^*,w)),R) \coloneqq \One[K_j=k_j]$. This gives
\begin{equation}
\label{eq:reconditioning1}
\begin{aligned}
&\Pr_{v^*,w,R}[v^* = v \mid K_j = k_j]\\
&\qquad=\Pr_{v^*,w,R}[v^*=v \mid \underbrace{v^* \in [q_{l_j},q_{r_j}]}_\text{(c)} \wedge \underbrace{\dw(v^*,q_i) = a_i \,\forall i \geq j}_\text{(b)} \wedge K_j= k_j]\\
&\qquad=\Pr_{v^*,w,R}[\underbrace{v^*=v}_E \mid \underbrace{v^* \in [q_{l_j},q_{r_j}] \wedge \dw(v^*,q_i)=a_i \,\forall i \geq j}_F \wedge \underbrace{K_j= k_j}_G]\\
&\qquad=\Pr_{v^*,w}[\underbrace{v^*=v}_E \mid \underbrace{v^* \in [q_{l_j},q_{r_j}] \wedge \dw(v^*,q_i)=a_i \,\forall i \geq j}_F]
\end{aligned}
\end{equation}

Now, let us conceptually split $w$ into two parts: $\wIn$, which contains the weights of only the edges between nodes $q_{l_j}$ and $q_{r_j}$, and $\wOut$, which contains all the other weights. Since weights are distributed independently, $(v^*,\wIn)$ is independent from $\wOut$.

Let $E((v^*,\wIn)) \coloneqq\One[v^*=v]$, $F((v^*,\wIn))\coloneqq \One[v^* \in [q_{l_j},q_{r_j}] \wedge \dw(v^*,q_{l_j}) = a_{l_j} \wedge \dw(v^*,q_{r_j}) = a_{r_j}]$, and $G(F((v^*,\wIn)),\wOut) \coloneqq \One[\dw(v^*,q_i)=a_i \,\forall i \geq j]$. It is easy to see why $E$ and $F$ depend only on $v^*$ and $\wIn$. For $G$, we note that all other answers can be deduced just from the fact that $v^* \in [q_{l_j},q_{r_j}]$, the values of $\dw(v^*,q_{l_j})$ and $\dw(v^*,q_{r_j})$, and $\wOut$. Therefore, $G$ depends only on $F$ and $\wOut$. Thus we can again apply Fact~\ref{fact:conditioning}, to obtain
\begin{equation}\label{eq:reconditioning2}
\begin{aligned}
&\Pr_{v^*,w}[v^*=v \mid v^* \in [q_{l_j},q_{r_j}] \wedge \dw(v^*,q_i)=a_i \,\forall i \geq j]\\
&\qquad= \Pr_{v^*,w}[\underbrace{v^* = v}_E \mid \underbrace{v^* \in [q_{l_j},q_{r_j}] \wedge \dw(v^*,q_{l_j}) = a_{l_j} \wedge \dw(v^*,q_{r_j}) = a_{r_j}}_F \wedge \underbrace{\dw(v^*,q_i)=a_i \,\forall i \geq j}_G]\\
&\qquad= \Pr_{v^*,w}[\underbrace{v^* = v}_E \mid \underbrace{v^* \in [q_{l_j},q_{r_j}] \wedge \dw(v^*,q_{l_j}) = a_{l_j} \wedge \dw(v^*,q_{r_j}) = a_{r_j}}_F]\\
&\qquad= \Pr_{v^*,w}[v^* = v \mid \GroupComplex].
\end{aligned}
\end{equation}

The result follows from combining \eqref{eq:reconditioning1} with \eqref{eq:reconditioning2}.
\end{proof}

For the remainder of this section, to make the notation lighter, we will use the following shorthands.
\[l \coloneqq q_{l_j} \quad r \coloneqq q_{r_j} \quad a_l \coloneqq a_{l_j} \quad a_r \coloneqq a_{r_j} \quad \mu \coloneqq \mu_j = \min(a_l, a_r)\]

%%%%%%%%%%%%%%%%%%%%%%%%%
%%%%%% DEFINITIONS %%%%%% (mu and I)
%%%%%%%%%%%%%%%%%%%%%%%%%

\begin{definition}[$\mup$]
Let $\mup \coloneqq \frac{\ssmu}{400\ln \mu}$.
\end{definition}

\begin{definition}[$I$]
Let $I$ be the interval of all sources $v^*$ that are consistent with event $T$ and $\Group{}$: more precisely,
\[I \coloneqq \{v \mid \Pr[T \wedge (v^*=v) \mid \Group] > 0\}.\]
\end{definition}

%%%%%%%%%%%%%%%%%%%
%%%%%% LEMMA %%%%%% (conditional probability is diluted)
%%%%%%%%%%%%%%%%%%%

\begin{lemma}\label{lemma:diluted}
Assume $\mu,\mup \geq D$. Then for any node $v \in I$,
\[\Pr_{v^*,w}[v^* = v \mid \Group{}] \leq \frac{100\ln \mu}{\ssmu}.\]
\end{lemma}

Lemma~\ref{lemma:diluted} has long and complicated proof, but its meaning is intuitive: it states is that when $\mu$ is large, the posteriors of $v^*$ are not very concentrated at any point of segment $I$. The expression of this posterior is too complex to work with directly, so we will need the help of some facts and claims to prove what we want.

Before we prove Lemma~\ref{lemma:diluted}, we start with Fact~\ref{fact:ineqs}, which gives us a couple of useful inequalities, and Fact~\ref{fact:decr-incr}, a simple calculus result that we will use in this section and the next.

%%%%%%%%%%%%%%%%%%
%%%%%% FACT %%%%%% (inequalities)
%%%%%%%%%%%%%%%%%%

\begin{fact}\label{fact:ineqs}
Assume $\mu,\mup \geq D$. Then
\begin{equation}\label{eq:mu-is-very-big}
\mu \geq 800^2\cdot \max(\mup, D, \sigma^2, 1).
\end{equation}
In particular, this implies
\begin{align}
\mu &\geq (5/4)D\label{eq:ineq1}\\
\ssmu &\leq \mu/5\label{eq:ineq2}\\
\ssmu &\leq \sigma \sqrt{(4/5)\mu}\,\ln((4/5)\mu)\label{eq:ineq3}\\
(3/5)\mu &\geq 6000.\label{eq:ineq4}
\end{align}
\end{fact}

\begin{proof}
Recall from equation~\eqref{eq:guarantees-on-cd} that $D \geq \sigma^2, e^2$.
Since $\mu \geq D \geq e^2$, we have
\[400\ln \mu \geq 400\ln e^2 = 800.\]
This implies
\[\sigma^2 \leq D \leq \mup = \frac{\ssmu}{400\ln \mu} \leq \frac{\ssmu}{800},\]
from which we get $\sigma \leq \sqrt{\mu}/800$ and $\mup \leq \frac{\sigma\sqrt{\mu}}{800}$. Combining those, we get
\[\mup \leq \frac{\sigma\sqrt{\mu}}{800} \leq \frac{\mu}{800^2},\]
which proves $\mu \geq 800^2\cmup$.
The other three parts then follow directly from $\mup \geq D \geq \sigma^2, e^2$.

Among \eqref{eq:ineq1}--\eqref{eq:ineq4}, all are trivial from \eqref{eq:mu-is-very-big}, except for \eqref{eq:ineq3} which can be rewritten as $5/4 \leq \ln((4/5)\mu)$. This clearly holds for $\mu \geq 10^4$.
\end{proof}

%%%%%%%%%%%%%%%%%%
%%%%%% FACT %%%%%% (reduce is increasing)
%%%%%%%%%%%%%%%%%%

\begin{fact}\label{fact:decr-incr}
Let $f(d) \coloneqq \frac{\ln d}{\sqrt{d}}$ and $g(d) \coloneqq \frac{\sqrt{d}}{\ln d}$. On $[e^2,\infty)$, $f$ is decreasing and $g$ is increasing.
\end{fact}

\begin{proof}
The derivative of $f$ is $\frac{2-\ln d}{2d\sqrt{d}}$.
\end{proof}

% PROOF OF THE LEMMA

\begin{proof}[Proof of Lemma~\ref{lemma:diluted}]
Since we have $\mu,\mup \geq D$, Fact~\ref{fact:ineqs} applies here.

Assume $I$ is not empty (otherwise, the lemma holds vacuously). Let $I'$ be $I$ extended by $\ssmu$ on both sides. Note that, by definition, the length of $I'$ is at least $2\ssmu$.

%\addFig{drawing of $I$ and $I'$ inside of $[l,r]$}

\begin{claim}\label{claim:dl-dr-prime}
For any $v' \in I'$, let $d_l' \coloneqq v' - l$ and $d_r' \coloneqq r - v'$. Then
\begin{align}
d_l' &\in [1/2,2]a_l\label{eq:dlprime1}\\
d_r' &\in [1/2,2]a_r\label{eq:drprime1}\\
|a_l-d_l'| &\leq 3 \sigma\sqrt{d_l'}\ln d_l'\label{eq:dlprime2}\\
|a_r-d_r'| &\leq 3 \sigma\sqrt{d_r'}\ln d_r'.\label{eq:drprime2}
\end{align}
\end{claim}

\begin{claimproof}
We will only prove \eqref{eq:dlprime1} and \eqref{eq:dlprime2}; the proof of the other two is analogous.

Let us first take a look at the properties of source candidates in $I$. Take some candidate $v \in I$ and consider the (true) distance between $v$ and the left side of the interval, $d_l \coloneqq v-l$. First, let us show that $d_l \geq D$. Indeed, if $d_l < D$, then point $v-D$ would be strictly to the left of $l$. Thus we would have
\[\dw(v, v-D) > \dw(v, l) = a_l \geq \mu \stackrel{\text{\eqref{eq:ineq1}}}{\geq} (5/4)D,\]
which contradicts part (ii) in Definition~\ref{def:typical} for $q\coloneqq v-D$.% If that were the case, $v$ wouldn't be consistent with $\Typical{}$.

Now that we have $d_l \geq D$, we can apply part (ii) in Definition~\ref{def:typical} again to obtain that
\begin{equation}\label{eq:dl1}
a_l \in [3/4,5/4]d_l\ \Rightarrow\ d_l \in [4/5,4/3]a_l
\end{equation}
and
\begin{equation}\label{eq:dl2}
|a_l-d_l| \in \sigma\sqrt{d_l}\ln d_l.
\end{equation}

%\addFig{drawing with $v'$ and its corresponding $v$}

Let us now extend these results to $I'$. Any point $v' \in I'$ is at most $\ssmu$ away from some point $v \in I$. Therefore, if we continue using notation $d_l \coloneqq v-l$, we have $d_l' \in d_l \pm \ssmu$.

From \eqref{eq:dl1}, we get
\[d_l' \in d_l \pm \ssmu \subseteq [4/5,4/3]a_l \pm\ssmu.\]
Besides, from \eqref{eq:ineq2}, $\ssmu \leq \mu/5 \leq a_l/5$, so we have
\begin{equation*}%\label{eq:dlprime1-strong}
d_l' \in [4/5,4/3]a_l \pm a_l/5 = [3/5,23/15]a_l,
\end{equation*}
which proves \eqref{eq:dlprime1}.

In addition, we note that
\begin{equation}\label{eq:dlprime-dl}
d_l' \geq d_l - \ssmu \geq d_l - a_l/5 \stackrel{\eqref{eq:dl1}}{\geq} d_l - d_l/4 = (3/4)d_l.
\end{equation}
Therefore,
\begin{align*}
|a_l-d_l'|
&\leq |a_l-d_l| + \ssmu\tag{definition of $I'$}\\
&\leq \sigma\sqrt{d_l} \ln d_l + \ssmu\tag{from \eqref{eq:dl2}}\\
&\leq \sigma\sqrt{d_l} \ln d_l + \sigma \sqrt{(4/5)\mu}\ln((4/5)\mu)\tag{from \eqref{eq:ineq3}}\\
&\leq 2\sigma \sqrt{d_l} \ln d_l\tag{$d_l \geq (4/5)a_l \geq (4/5)\mu$}\\
&\leq 2\sigma \sqrt{4d_l'/3} \ln (4d_l'/3)\tag{from \eqref{eq:dlprime-dl}}\\
&\leq 3\sigma \sqrt{d_l'} \ln d_l',
\end{align*}
where the last step holds because $d_l' \geq (3/5)a_l \geq (3/5)\mu \stackrel{\eqref{eq:ineq4}}{\geq} 6000$, which is big enough. This proves \eqref{eq:dlprime2}.
\end{claimproof}

Note that Claim~\ref{claim:dl-dr-prime} implies in particular that $I' \subseteq [l,r]$. Let us study the ratios of the posterior probabilities of $v^*=v'$ between different values of $v' \in I'$, conditioned on $a_l, a_r$ (but regardless of whether $T$ holds). We will use the shorthand
\begin{equation}\label{eq:definition-pvprime}
p(v') \coloneqq \Pr[v^*=v' \mid v^* \in I' \wedge \dw(v^*,l)=a_l \wedge \dw(v^*,r)=a_r].
\end{equation}
Given that $v^*$ is initially distributed uniformly, $p(v')$ is proportional to
\begin{equation}
\begin{aligned}\label{eq:separate-lr}
&\Pr[\dw(v^*,l)=a_l \wedge \dw(v^*,r)=a_r\mid v^*=v']\\
&\qquad = \Pr[\dw(v^*,l)=a_l\mid v^*=v']\Pr[\dw(v^*,r)=a_r\mid v^*=v'],
\end{aligned}
\end{equation}
where the independence comes from the fact that $\dw(v^*,l)$ and $\dw(v^*,r)$ depend on completely separate weights.

Because we assumed that all weights are independently distributed from $\Normal(1,\sigma^2)$, both factors in \eqref{eq:separate-lr} follow a normal distribution. Those distributions are $\Normal(d_l',\sigma^2d_l')$ and $\Normal(d_r',\sigma^2d_r')$, where we continue notations $d_l'\coloneqq v'-l$ and $d_r'\coloneqq r-v'$.
This means that $p(v')$ is proportional to
\begin{equation*}%\label{eq:raw-proba}
\frac{1}{\sqrt{2\pi\sigma^2d_l'}}e^{-\frac{(a_l-d_l')^2}{2\sigma^2d_l'}} \times \frac{1}{\sqrt{2\pi\sigma^2d_r'}}e^{-\frac{(a_r-d_r')^2}{2\sigma^2d_r'}}.
\end{equation*}
Note that in the above expression, $a_l,a_r$ are fixed by the conditioning, while $d_l'$ and $d_r'$ depend on $v'$. From now on, we will denote them as $\myDl$ and $\myDr$ to make this clear.

Of course, the constant $\frac{1}{2\pi\sigma^2}$ does not matter. Besides, we know that $\myDl \in [1/2,2]a_l$ and $\myDr \in [1/2,2]a_r$ with $a_l,a_r$ fixed, so the factor $\frac{1}{\sqrt{\myDl\myDr}}$ will vary only by a factor 4. Thus we can conclude that $p(v')$ is also proportional to
\begin{equation*}%\label{eq:proba-without-scaling}
F(v') \coloneqq e^{-\frac{1}{2\sigma^2}\left(\frac{(a_l-\myDl)^2}{\myDl} + \frac{(a_r-\myDr)^2}{\myDr}\right)},
\end{equation*}
up to a factor 4 of error. More precisely, we know that there exists some $k>0$ such that for all $v' \in I'$,
\begin{equation}\label{eq:proportionality-to-f}
p(v') \in [1,4]kF(v').
\end{equation}

Let
\[
G(v') \coloneqq \frac{1}{2\sigma^2}\left(\frac{(a_l-\myDl)^2}{\myDl} + \frac{(a_r-\myDr)^2}{\myDr}\right)
\]
so that $F(v') = e^{-G(v')}$, and let $\vpeak$ be the value of $v' \in I'$ that maximizes the expression $F(v')$. We will show the existence of a relatively large interval $J \subseteq I'$ centered around $\vpeak$ such that for all $v' \in J$, the value $F(v')$ is not much smaller than $F(\vpeak)$.

\begin{claim}\label{claim:interval-j}
There is some interval $J \subseteq I'$ of length at least $\frac{\ssmu}{6\sqrt{2}\ln \mu}$, such that for all $v' \in J$,
\[F(v') \geq \frac{F(\vpeak)}{e}.\]
\end{claim}
\begin{proof}
The derivative of $F(v')$ is $G'(v')e^{G(v')} = G'(v')F(v')$, where
\[
G'(v') = \frac{1}{2\sigma^2}\left(-\frac{a_l-\myDl}{\myDl} - \left(\frac{a_l-\myDl}{\myDl}\right)^2 + \frac{a_r-\myDr}{\myDr} + \left(\frac{a_r-\myDr}{\myDr}\right)^2\right).
\]
First, note that by \eqref{eq:dlprime1} and \eqref{eq:drprime1}, we have
\[\left|\frac{a_l-\myDl}{\myDl}\right| \leq 1\text{ and }\left|\frac{a_r-\myDr}{\myDr}\right| \leq 1.\]
Therefore,
\[|G'(v')| \leq \frac{1}{\sigma^2}\left(\left|\frac{a_l-\myDl}{\myDl}\right| + \left|\frac{a_r-\myDr}{\myDr}\right|\right).\]
And then we can use equations \eqref{eq:dlprime2} and \eqref{eq:drprime2} to bound this further:
\begin{align*}
|G'(v')| &\leq \frac{1}{\sigma^2}\left(\frac{3\sigma \sqrt{\myDl} \ln \myDl}{\myDl} + \frac{3\sigma \sqrt{\myDr} \ln \myDr}{\myDr}\right)\\
&\leq \frac{3}{\sigma}\left(\frac{\ln \myDl}{\sqrt{\myDl}} + \frac{\ln \myDr}{\sqrt{\myDr}}\right)\\
&\leq \frac{3}{\sigma}\left(\frac{\ln (\mu/2)}{\sqrt{\mu/2}} + \frac{\ln (\mu/2)}{\sqrt{\mu/2}}\right)\tag{$d_l' \geq a_l/2 \geq \mu/2 \geq e^2$ and Fact~\ref{fact:decr-incr}; same for $d_r'$}\\
&\leq \frac{6\sqrt{2}\ln \mu}{\sigma\sqrt{\mu}}.
\end{align*}
%\addFig{drawing of $F(v')$ and its lower bound around $\vpeak$}
We now have
\[|F'(v')| \leq \frac{6\sqrt{2}\ln \mu}{\ssmu} F(v'),\]
for all $v' \in I'$, which from $F(\vpeak) > 0$ and Grönwall's Lemma for ordinary differential inequalities can be seen to imply that for all $v' \in I'$,
\[F(v') \geq e^{-\left|v'-\vpeak\right|\frac{6\sqrt{2}\ln\mu}{\ssmu}} F(\vpeak).\]
This means that for any $v' \in I'$ within distance at most $\frac{\ssmu}{6\sqrt{2}\ln\mu}$ of $\vpeak$,
\[F(v') \geq e^{-1}F(\vpeak) \geq \frac{F(\vpeak)}{e}.\]
We can then set
\[J \coloneqq I' \cap \left(\vpeak \pm \frac{\ssmu}{6\sqrt{2}\ln\mu}\right).\]
Defined this way, $J$ will clearly have length at least $\frac{\ssmu}{6\sqrt{2}\ln\mu}$. Indeed $\vpeak$ is in $I'$, and the length of $I'$ is at least
$2\ssmu \geq 2 \times \frac{\ssmu}{6\sqrt{2}\ln\mu}$.
\end{proof}
Now, recall from \eqref{eq:proportionality-to-f} that
$p(v') \in [1,4]kF(v')$.
Thus for any $v' \in I'$,
\begin{equation}\label{eq:p-ub}
p(v') \leq 4kF(v') \leq 4kF(\vpeak).
\end{equation}
But in particular, by Claim~\ref{claim:interval-j}, for any $v' \in J$,
\begin{equation}\label{eq:p-lb}
p(v') \geq kF(v') \geq \frac{kF(\vpeak)}{e}.
\end{equation}
Besides, being a probability distribution, $p(v')$ must sum up to 1, so we have
\[1 = \sum_{v' \in I'} p(v') \geq \sum_{v' \in J} p(v') \stackrel{\eqref{eq:p-lb}}{\geq} |J| \frac{kF(\vpeak)}{e},\]
thus $kF(\vpeak) \leq \frac{e}{|J|}$. Therefore, for any $v' \in I'$,
\begin{equation}\label{eq:ub-pvprime}
p(v')
\stackrel{\eqref{eq:p-ub}}{\leq} 4kF(\vpeak)
\leq \frac{4e}{|J|}
\leq \frac{24e\sqrt{2}\ln \mu}{\ssmu}
\leq \frac{100\ln \mu}{\ssmu}.
\end{equation}

We are finally ready to prove the lemma. For $v \in I$, we can observe that
\begin{align*}
\Pr[v^* &= v \mid \Group]\\
&= \Pr[v^* = v \mid v^* \in [l,r] \wedge \dw(v^*,l)=a_l \wedge \dw(v^*,r)=a_r]\\
&\leq \Pr[v^* = v \mid v^* \in I' \wedge \dw(v^*,l)=a_l \wedge \dw(v^*,r)=a_r]\tag{strengthen the condition}\\
&= p(v)\tag{defined in \eqref{eq:definition-pvprime}}\\
&\leq \frac{100\ln \mu}{\ssmu}.\tag{$v \in I \subset I'$ and \eqref{eq:ub-pvprime}}
\end{align*}
\end{proof}

With Lemma~\ref{lemma:diluted} in hand, we are finally ready to prove Lemma~\ref{lemma:a-implies-b}.
\begin{proof}[Proof of Lemma~\ref{lemma:a-implies-b}]
First, we show that $j < \jStop$ and $A_j$ imply that $\mu,\mup \geq D$. Indeed, $j < \jStop$ implies $\lambda_j,\lambda_{j+1} \geq D$. If in addition $A_j$ holds, then $\mu = \mu_j \geq \lambda_j \geq D$, and
\[\mup = \frac{\ssmu}{400\ln \mu} \geq \frac{\sigma \sqrt{\lambda_j}}{400\ln \lambda_j} > \frac{\sigma \sqrt{\lambda_j}}{400\ln \lambda_j \log n} = \lambda_{j+1} \geq D.\]
Therefore, the assumptions of Lemma~\ref{lemma:diluted} hold.

We need to prove that for any $v \in V$,
\[\Pr[T \wedge (v^*=v) \mid K_j] \leq \dilution{1}.\]
If $v \notin I$, then by definition of $I$, $\Pr[T \wedge (v^*=v)\mid K_j]=0$, so the inequality holds trivially. On the other hand, if $v \in I$,
\begin{align*}
\Pr[T \wedge (v^*=v) &\mid K_j]\\
&\leq \Pr[v^* = v \mid K_j]\\
&= \Pr[v^* = v \mid \Group]\tag{by Lemma~\ref{lemma:only-lr}}\\
&\leq \frac{100\ln \mu}{\ssmu}\tag{by Lemma~\ref{lemma:diluted}}\\
&\leq \frac{100\ln \lambda_j}{\sigma \sqrt{\lambda_j}}\\
&= \frac{1}{4\lambda_{j+1} \log n}\\
&\leq \dilution{1}.\tag{$\lambda_{j+1} \geq D \geq e^2$}
\end{align*}
\end{proof}
% END \input{proofs/a-implies-b.tex}
%\input{proofs/b-implies-a}
% BEGIN \input{proofs/asymptotics-lb.tex}
\section{Proof of Lemma~\ref{lemma:asymptotics-lb}}
\label{proof:lemma:asymptotics-lb}

\begin{replemma}{lemma:asymptotics-lb}
For $n \geq \Theta_p(\max(\sigma^3, 1))$, we have
\[
\jStop+1 =
\begin{cases}
\Omega_{p}(1 + \log (1+\log_{1/\sigma} n))\text{ if $\sigma^2 \leq 1/2$}\\
\Omega_{p}(\log \log n)\text{ if $\sigma^2 \geq 1/2$.}
\end{cases}
\]
\end{replemma}

This proof is very similar in spirit to the proof of Claim~\ref{claim:asymptotics-ub} (\ref{proof:claim:asymptotics-ub}).
\begin{proof}[Proof of Lemma~\ref{lemma:asymptotics-lb}]
We track the value of $\lambda_j/\sigma^2$ as $j$ increases. First, as long as
\begin{equation}
\label{eq:lambda-large-enough}
\frac{\lambda_j}{(\ln\lambda_j)^6} \geq \sigma^2(400\log n)^6,
\end{equation}
we have
\[
\frac{\lambda_{j+1}}{\sigma^2}
= \frac{\reduce(\lambda_{j+1})}{\sigma^2}
= \frac{\sigma\sqrt{\lambda_j}}{400\sigma^2\ln\lambda_j\log n}
= \sqrt{\frac{\lambda_j}{\sigma^2}}
\stackrel{\eqref{eq:lambda-large-enough}}{\geq}
\left(\frac{\lambda_j}{\sigma^2}\right)^{1/3}.
\]
Let $\lMin$ be the smallest possible value for $\lambda_j$ at least as large as $D$ such that \eqref{eq:lambda-large-enough} holds.
Then, by induction, as long as $\lambda_j \geq \lMin$, we have
\begin{equation}
\label{eq:lambda-lb}
\frac{\lambda_j}{\sigma^2} \geq \left(\frac{\lambda_0}{\sigma^2}\right)^{1/3^j} \Rightarrow \lambda_j \geq \sigma^2 \left(\frac{n}{C\sigma^2}\right)^{1/3^j}.
\end{equation}

Recall that $\jMin$ is the smallest integer $j\geq 0$ such that $\lambda_j < D$. Let $j^*$ be the smallest integer $j\geq 0$ such that $\lambda_j < \lMin$. Then $j^* \leq \jMin$, and
applying \eqref{eq:lambda-lb} to $j^*$ we get
\begin{equation}
\label{eq:lb-jmin}
\sigma^2 \left(\frac{n}{C\sigma^2}\right)^{1/3^{j^*}} < D \Rightarrow \jMin \geq j^* > \log_3\left(\frac{\log\left(\frac{n}{C\sigma^2}\right)}{\log\left(\frac{\lMin}{\sigma^2}\right)}\right).
\end{equation}

Since $\lambda_0 = n/C$, for $n \geq CD = O_{p}(\sigma^2\ln \sigma)$, we have $\jMin \geq 1$. We separate into cases to obtain more lower bounds of $\jMin$. First, if $\sigma^2 \leq 1/2$, it is easy to verify that $\lMin = \max(D,(\log n)^{O(1)}) = (\log n)^{O_{p}(1)}$.
Therefore, by \eqref{eq:lb-jmin}, there exists $n_1 = O_p(1)$ such that for $n \geq n_1$, we have
\[\jMin \geq \Omega_{p}(\log(1+\log_{1/\sigma}n)).\] Second, if $\sigma^2 \geq 1/2$, it is easy to verify that $\lMin = \max(D,\sigma^2(\log n)^{O(1)}) = \sigma^2(\log n)^{O_{p}(1)}$.
Therefore, by \eqref{eq:lb-jmin}, there exists $n_2 = O_p(\sigma^3)$ such that for $n \geq n_2$, we have
\[\jMin \geq \Omega_{p}(\log \log n).\]

In summary, there exists $n_3 = \max(n_1, n_2) = O_p(\max(\sigma^3, 1))$ such that for $n \geq n_3$,
\[
\jMin =
\begin{cases}
\Omega_{{p}}(1 + \log (1+\log_{1/\sigma} n))\text{ if $\sigma^2 \leq 1/2$}\\
\Omega_{{p}}(\log \log n)\text{ if $\sigma^2 \geq 1/2$,}
\end{cases}
\]
which is exactly the bounds we want for $\jStop + 1$.
We can then conclude by observing that
\[\jStop+1 = \min\left(\jMin, 1+\floor*{\frac{{p}\log n}{2}}\right)\]
and that $1+\floor*{\frac{{p}\log n}{2}}$ is larger than the claimed lower bounds for $n$ large enough (which is covered in the $\Theta_p(\max(\sigma^3,1))$ lower bound on $n$ in the statement of the lemma).
\end{proof}% END \input{proofs/asymptotics-lb.tex}
% END \input{annoying-proofs.tex}

\end{document}